\documentclass[a4paper,english]{lipics}

\usepackage[utf8]{inputenc}
\usepackage[noadjust,nosort]{cite}
\usepackage{amsmath}
\usepackage{amsfonts}
\usepackage{MnSymbol}
\usepackage{thmtools}
\usepackage{thm-restate}
\usepackage{booktabs}
\usepackage{figlatex}
\usepackage{xcolor}
\usepackage{setspace}
\usepackage{xspace}
\usepackage[super]{nth}
\usepackage{multicol}
\usepackage{wrapfig}
\usepackage{hyperref}
\usepackage[linesnumbered,lined,noend,ruled]{algorithm2e}
\usepackage[capitalise,english,nameinlink]{cleveref} 

\newcommand\rd    \propto

\newcommand\cfl      {\mathrel{\#}}
\newcommand\icfl[1][]{\mathrel{\#^i_{#1}}}

\newcommand\indep    {\mathrel{\meddiamond}}
\newcommand\depen    {\mathrel{\diamondtimes}}
\newcommand\eqdef    {\mathrel{:=}}

\newcommand\ispref   {\mathrel{\trianglelefteq}}

\newcommand\fire[1]  {\mathrel{\raisebox{-1.1pt}{$\xrightarrow{#1}$}}}

\newcommand\iscutoff[1] {\mathop{\mathsf{cutoff}} (#1)}
\newcommand\union[1]    {\mathop{\mathit{union}} (#1)}

\newcommand\bigo[1]     {\mathop{\mathcal{O}} (#1)}

\newcommand\conf[1]     {\mathop{\mathit{conf}} (#1)}

\newcommand\reach[1]    {\mathop{\mathit{reach}} (#1)}
\newcommand\runs[1]     {\mathop{\mathit{runs}} (#1)}

\newcommand\state[1]    {\mathop{\mathit{state}} (#1)}

\newcommand\inter[1]    {\mathop{\mathit{inter}} (#1)}

\newcommand\enabl[1]    {\mathop{\mathit{enabl}} (#1)}
\newcommand\en[1]       {\mathop{\mathit{en}} (#1)}
\newcommand\ex[1]       {\mathop{\mathit{ex}} (#1)}
\newcommand\cex[1]      {\mathop{\mathit{cex}} (#1)}
\newcommand\fcfl[1]     {\mathop{\mathit{\#}} (#1)}
\newcommand\ficfl[2][]  {\mathop{\mathit{\#^i_{#1} }} (#2)}

\newcommand\post[1]     {#1^\bullet}
\newcommand\pre[1]      {{}^\bullet#1}

\newcommand\causes[1]   {\left \lceil #1 \right \rceil}

\newcommand\les         {\mathcal{E}}
\newcommand\unf[1]      {\mathcal{U}_{#1}}

\newcommand\trace[1]    {\mathcal{T}_{#1}}
\newcommand\dgraph[1]   {\mathcal{D}_{#1}}
\newcommand\intsem[1]   {\mathcal{S}_{#1}}

\newcommand\poet     {\textsc{Poet}\@\xspace}
\newcommand\nidhugg  {\textsc{Nidhugg}\@\xspace}

\newcommand\uunf     {\mathcal{U}}
\newcommand\ppref    {\mathcal{P}}

\newcommand\N        {\mathbb{N}}

\newcommand\wlogg    {w.l.o.g.\xspace}

\newcommand\wrt      {w.r.t.\@\xspace}
\newcommand\cf       {cf.\@\xspace}
\newcommand\Wlog     {W.l.o.g.\xspace}
\newcommand\al       {al.\@\xspace}
\newcommand\eg       {e.g.\@\xspace}
\newcommand\etc      {etc.\@\xspace}
\newcommand\ie       {i.e.\@\xspace}

\newcommand\resp     {resp.\@\xspace}
\newcommand\nr       {nr.\@\xspace}
\newcommand\avg      {avg.\@\xspace}

\newcommand\cod[1]   {\texttt{#1}}
\newcommand\tup[1]   {\langle#1\rangle}
\newcommand\set[1]   {{\{ #1 \mathclose \}}}
\newcommand\eqtag[1] {\hfill{} \refstepcounter{equation} \label{#1} (\arabic{equation})}

\crefname{equation}{}{}
\crefname{proposition}{Prop.}{Props.}
\Crefname{proposition}{Proposition}{Propsitions}
\crefname{section}{\S}{\S\S}
\Crefname{section}{Section}{Sections}
\crefname{page}{p.}{pp.}
\Crefname{page}{Page}{Pages}
\crefname{chapter}{Ch.}{Ch.}
\Crefname{chapter}{Chapter}{Chapters}
\crefname{remark}{Rmk.}{Rmks.}
\Crefname{remark}{Remark}{Remarks}
\crefname{appendix}{App.}{App.}
\Crefname{appendix}{Appendix}{Appendix}
\crefname{algorithm}{Alg.}{Alg.}
\Crefname{algorithm}{Algorithm}{Algorithms}
\crefname{line}{line}{line}
\Crefname{line}{Line}{Line}

\crefname{theo}{Theorem}{Theorems}
\crefname{lem}{Lemma}{Lemmas}
\crefname{cor}{Corollary}{Corollaries}
\crefname{prop}{Prop.}{Props.}
\crefname{defn}{Def.}{Defs.}
\Crefname{defn}{Definition}{Definitions}
\crefname{exampl}{Example}{Examples}
\crefname{rmk}{Remark}{Remarks}

\hypersetup{
  bookmarksdepth=2,
  bookmarksnumbered=true,
  bookmarksopen=true,
  bookmarksopenlevel=2,
  colorlinks=true,
  linktocpage=true,
  breaklinks=true,
  pageanchor=true,
  allcolors=[rgb]{0.6,0.0,0.0},
  pdftitle={Unfolding-based Partial Order Reduction},
  pdfauthor={César Rodríguez, Marcelo Sousa, Subodh Sharma, Daniel Kroening}
}

\theoremstyle{plain}

\newtheorem{lem}[theorem]{Lemma}
\newtheorem{cor}[theorem]{Corollary}
\newtheorem{prop}[theorem]{Proposition}
\theoremstyle{definition}
\newtheorem{defn}[theorem]{Definition}
\newtheorem{exampl}[theorem]{Example}
\theoremstyle{remark}

\newcommand{\remove}[1]{}

\ifdefined\VersionWithComments
\newcommand{\cro}[1]{{\color{teal}\textbf{[César}: #1]}}
\newcommand{\ms}[1]{{\color{blue}\textbf{[Marcelo}: #1]}}
\newcommand{\svs}[1]{{\color{orange}\textbf{[Subodh}: #1]}}
\newcommand{\todo}[1]{{\color{red!90!black}\textbf{[TODO}: #1]}}
\else
\newcommand{\cro}[1]{}
\newcommand{\ms}[1]{}
\newcommand{\svs}[1]{}
\newcommand{\todo}[1]{}
\fi

\title{Unfolding-based Partial Order Reduction\footnote{This research was
supported by ERC project~280053 (CPROVER).\newline
This is the long version of a paper of the same title apparead at the
proceedings of CONCUR 2015.}}

\author[1]{César Rodríguez}
\author[2]{Marcelo Sousa}
\author[3]{Subodh Sharma}
\author[4]{\\Daniel Kroening}

\affil[1]{Université Paris 13, Sorbonne Paris Cité, LIPN, CNRS, France}
\affil[2,4]{Department of Computer Science, University of Oxford, UK}
\affil[3]{Indian Institute of Technology Delhi, India}

\authorrunning{C.\ Rodríguez et \al}
\Copyright{César Rodríguez, Marcelo Sousa, Subodh Sharma, Daniel Kroening}

\serieslogo{}
\volumeinfo
  {}
  {0}
  {}
  {1}
  {1}
  {1}
\EventShortName{}

\begin{document}

\maketitle

\begin{abstract}
Partial order reduction (POR) and net unfoldings are two alternative methods
to tackle state-space explosion caused by concurrency.  In this paper, we
propose the combination of both approaches in an effort to combine their
strengths.  We first define, for an abstract execution model, unfolding
semantics parameterized over an arbitrary independence relation.  Based on
it, our main contribution is a novel stateless POR algorithm that explores
at most one execution per Mazurkiewicz trace, and in general, can explore
exponentially fewer, thus achieving a form of \emph{super-optimality}. 
Furthermore, our unfolding-based POR copes with non-terminating executions
and incorporates \mbox{state-caching}.  Over benchmarks with busy-waits,
among others, our experiments show a dramatic reduction in the number of
executions when compared to a state-of-the-art DPOR.
\end{abstract}

\section{Introduction}
\label{s:intro}

Efficient exploration of the state space of a concurrent system is a
fundamental problem in automated verification.
Concurrent actions often interleave in intractably many ways,
quickly populating the state space with many equivalent but unequal states.
Existing approaches to address this
can essentially be classified as either
partial-order reduction techniques (PORs) or unfolding methods.

POR methods~\cite{Val91,God96,FG05,GFYS07,YWY06,YCGK08,AAJS14,AAJLS15}
conceptually exploit the fact that executing certain transitions can be
postponed owing to their result being independent of the execution sequence
taken in their stead.
They execute a provably-sufficient subset of transitions enabled at
every state, computed either statically~\cite{Val91,God96}
or dynamically~\cite{FG05,AAJS14}.
The latter methods, referred as dynamic PORs (DPORs),
are often \emph{stateless} (\ie, they only store one execution in memory)
and constitute the most promising algorithms of the family.
By contrast, unfolding approaches~\cite{Mcm93,ERV02,BHKTV14,KH14}
model execution by partial orders, bound together by a conflict relation.
They construct finite, complete prefixes by a saturation procedure,
and cope with non-terminating executions using cutoff events~\cite{ERV02,BHKTV14}.

While a POR can employ arbitrarily sophisticated decision procedures to
choose a sufficient subset of transitions to fire,
in most cases~\cite{God96,FG05,GFYS07,YWY06,YCGK08,AAJS14,AAJLS15}
the \emph{commutativity of transitions} is the
enabling mechanism underlying the chosen procedure.
Commutativity, or independence, is thus a mechanism and not necessarily an
irreplaceable component of a POR~\cite{Val91,HW11}.\footnote{Though it is a very
popular one, all PORs based on persistent sets~\cite{God96},
for instance, are based on commutativity.}
PORs that exploit such commutativity conceptually
establish an equivalence relation on the sequential executions of the system
and explore at least one representative of each class, thus discarding equivalent
executions.
In this work we restrict our attention to exclusively PORs that exploit commutativity.

Despite impressive advances in the field,
both unfoldings and PORs have shortcomings.
We now give six of them.
Current unfolding algorithms
(1)~need to solve an \mbox{NP-complete} problem when adding events to the
unfolding~\cite{Mcm93}, which seriously limits the performance of existing
unfolders as the structure grows.
They are also (2)~inherently \emph{stateful},
\ie, they cannot selectively discard visited events from memory,
quickly running out of it.
PORs, on the other hand, explore Mazurkiewicz traces~\cite{Maz87}, which
(3)~often outnumber the events in the corresponding unfolding by an
exponential factor (\eg, \cref{f:2read}~(d) gives an unfolding
with~$2n$ events and~$\bigo{2^n}$ traces).
Furthermore, DPORs often
(4)~explore the same states repeatedly~\cite{YCGK08},
and combining them with stateful search,
although achieved for non-optimal DPOR~\cite{YCGK08,YWY06},
is difficult because of the dynamic nature of DPOR~\cite{YWY06}.
More on this in~\cref{m:overview}.
The same holds when extending DPORs to (5) cope with non-terminating executions
(note that a solution to (4) does not necessarily solve (5)).
Lastly, (6) existing stateless PORs do not exploit additional available
memory (RAM) for any other purpose.

Either readily available solutions or promising directions to address these six
problems can be found in, respectively, the opposite approach.
PORs inexpensively add events to the current execution,
contrary to unfoldings~(1).
They easily discard events from memory when backtracking, which addresses~(2).
On the other hand, while PORs explore Mazurkiewicz traces
(\emph{maximal configurations}), unfoldings explore events
(\emph{local configurations}), thus addressing~(3).
Explorations of repeated states and pruning of non-terminating executions
is elegantly achieved in unfoldings by means of cutoff events.
This solves~(4) and~(5).

Some of these solutions indeed seem, at present, incompatible with each other.
We do not mean that the combination of POR and unfoldings immediately
addresses the above problems.
However, since both unfoldings and PORs share many fundamental
similarities, tackling these problems in a unified framework is likely to
shed light on them.

This paper lays out a DPOR algorithm on top of an unfolding structure.
Our main result is a novel stateless, optimal DPOR that
explores at most once every Mazurkiewicz trace,
and often many fewer owing to cutoff events (cutoffs stop traces that could
later branch into multiple traces).
It also copes with non-terminating systems
and exploits all available RAM with a \emph{cache memory} of
events, speeding up revisiting events.
This provides a solution to (4), (5), (6), and a partial solution to~(3).
Our algorithm can alternatively be viewed as a stateless unfolding exploration,
partially addressing~(1) and~(2).

Our result reveals DPORs as algorithms exploring an
object that has richer structure than a plain directed graph.
Specifically, unfoldings provide a solid notion of event \emph{across multiple
executions}, and a clear notion of conflict.
Our algorithm indirectly maps important POR notions to concepts in unfolding theory.

\begin{exampl}
\label{m:overview}
We illustrate problems (3), (4), and (5), and show how our DPOR deals with them.
The following code is the skeleton of a \mbox{producer-consumer} program.
Two concurrent producers write in, \resp, \verb!buf1! and \verb!buf2!.
The consumer access the buffers in sequence.
\noindent
\begin{center}
\small
\begin{minipage}[t]{0.33\textwidth}
\begin{verbatim}
while (1):
  lock(m1)
  if (buf1 < MAX): buf1++
  unlock(m1)
\end{verbatim}
\end{minipage}%
\vrule~%
\begin{minipage}[t]{0.33\textwidth}
\begin{verbatim}
while (1):
  lock(m2)
  if (buf2 < MAX): buf2++
  unlock(m2)
\end{verbatim}
\end{minipage}%
\vrule~%
\begin{minipage}[t]{0.31\textwidth}
\begin{verbatim}
while (1):
  lock(m1)
  if (buf1 > MIN): buf1--
  unlock(m1)
  // same for m2, buf2
\end{verbatim}
\end{minipage}%
\end{center}
Lock and unlock operations on both mutexes \verb!m1! and \verb!m2! create many
Mazurkiewicz traces. However, most of them have isomorphic \emph{suffixes}, \eg,
producing two items in \verb!buf1! and consuming one reaches the same state as
only producing one. After the common state, both traces explore identical
behaviours and only one needs to be explored. We use cutoff events,
inherited from unfolding theory~\cite{ERV02,BHKTV14}, to dynamically
stop the first trace and continue only with the second.
This addresses (4) and (5),  and partially deals with (3).
Observe that cutoff events are a form of semantic pruning,
in contrast to the syntactic pruning introduced by, \eg, bounding the depth of
loops, a common technique for coping with non-terminating executions in DPOR.
With cutoffs, the exploration can build unreachability \emph{proofs},
while depth bounding renders DPOR incomplete, \ie, it can only
\emph{find bugs}.
\end{exampl}

Our first step is to formulate PORs and unfoldings in the same framework. 
PORs are often presented for abstract execution models, while unfoldings
have mostly been considered for Petri nets, where the definition is
entangled with the syntax of the net.  We make a second contribution here. 
We define, for a general execution model, event structure
semantics~\cite{NPW81} parametric on a given independence relation.

\Cref{s:por} sets up basic notions and \cref{s:unfsem} presents our
parametric event-structure semantics.  In \cref{s:unfalgo} we introduce our
DPOR, \cref{s:impro} improves it with cutoff detection and discusses event
caching.  Experimental results are in~\cref{s:experiments} and related work
in~\cref{s:related}.  We conclude in~\cref{s:concl}.  All lemmas cited along
the paper and proofs of all stated results can be found in the appendixes.

%

\section{Execution Model and Partial Order Reductions}
\label{s:por}

We set up notation and recall general notions about PORs.
We consider an abstract model of (concurrent) computation.
A \emph{system} is a tuple
$M \eqdef \tup{\Sigma, T, \tilde s}$ formed by
a set~$\Sigma$ of \emph{global states},
a set~$T$ of \emph{transitions} and some
\emph{initial global state} $\tilde s \in \Sigma$.
Each transition $t \colon \Sigma \to \Sigma$ in~$T$ is
a \emph{partial} function accounting for how the occurrence of~$t$
transforms the state of~$M$.

A transition $t \in T$ is \emph{enabled} at a state~$s$ if~$t(s)$ is defined.
Such~$t$ can \emph{fire} at~$s$, producing a new state $s' \eqdef t (s)$.
We let $\enabl s$ denote the set of transitions enabled at~$s$.
The \emph{interleaving semantics} of $M$ is the
directed, edge-labelled graph
$\intsem M \eqdef \tup{\Sigma, {\to}, \tilde s}$ where
$\Sigma$ are the global states,
$\tilde s$ is the initial state and
${\to} \subseteq \Sigma \times T \times \Sigma$
contains a triple $\tup{s, t, s'}$, denoted by $s \fire t s'$, iff~$t$ is
enabled at~$s$ and~$s' = t(s)$.
Given two states $s, s' \in \Sigma$, and
$\sigma \eqdef t_1.t_2 \ldots t_n \in T^*$
($t_1$~concatenated with~$t_2$, \ldots until~$t_n$),
we denote by $s \fire{\sigma} s'$
the fact that there exist states $s_1, \ldots, s_{n-1} \in \Sigma$ such that
$s \fire{t_1} s_1$, \ldots, $s_{n-1} \fire{t_n} s'$.

A \emph{run} (or \emph{interleaving}, or \emph{execution})
of $M$ is any sequence $\sigma \in T^*$ such that
$\tilde s \fire \sigma s$ for some $s \in \Sigma$.
We denote by~$\state \sigma$ the state~$s$ that $\sigma$ reaches,
and by~$\runs M$ the set of runs of~$M$, also referred to as the
\emph{interleaving space}.
A state~$s \in \Sigma$ is \emph{reachable}
if $s = \state \sigma$ for some $\sigma \in \runs M$;
it is a \emph{deadlock} if $\enabl s = \emptyset$, and in that
case~$\sigma$ is called \emph{deadlocking}.
We let $\reach M$ denote the set of reachable states in~$M$.
For the rest of the paper,
we fix a system~$M \eqdef \tup{\Sigma, T, \tilde s}$ and assume that
$\reach M$ is finite.


The core idea behind PORs\footnote{To be completely correct we should say
``\textit{PORs that exploit the independence of transitions}''.}
is that certain transitions can be
seen as commutative operators, \ie, changing their order of occurrence
does not change the result.
Given two transitions $t, t' \in T$ and one state $s \in \Sigma$,
we say that $t, t'$ \emph{commute at} $s$ iff
\begin{itemize}
\item
  if $t \in \enabl s$ and $s \fire t s'$,
  then $t' \in \enabl s$ iff $t' \in \enabl{s'}$; and
\item
  if $t, t' \in \enabl s$,
  then there is a state $s'$ such that
  $s \fire{t.t'} s'$ and
  $s \fire{t'.t} s'$.
\end{itemize}
For instance, the lock operations on \verb!m1! and \verb!m2!
(\cref{m:overview}), commute on every state, as they update different variables.
Commutativity of transitions at states identifies an equivalence relation
on the set~$\runs M$.
Two runs $\sigma$ and $\sigma'$ of the same length
are \emph{equivalent}, written $\sigma \equiv \sigma'$,
if they are the same sequence modulo swapping commutative transitions.
Thus equivalent runs reach the same state.
POR methods explore a fragment of $\intsem M$ that contains at least one run
in the equivalence class of each run that reaches each deadlock state.
This is achieved by means of a so-called~\emph{selective search}~\cite{God96}.
Since employing commutativity can be expensive, PORs often use
\emph{independence relations},
\ie, sound under-approximations of the commutativity relation.
In this work, partially to simplify presentation, we use unconditional
independence.

Formally,
an \emph{unconditional independence relation} on $M$ is any symmetric and
irreflexive relation ${\indep} \subseteq T \times T$ such that
if $t \indep t'$, then $t$ and $t'$ commute at \emph{every} state
$s \in \reach M$.
If $t, t'$ are not independent according to $\indep$, then they are
\emph{dependent}, denoted by $t \depen t'$.

Unconditional independence identifies an equivalence
relation $\equiv_{\indep}$ on the set $\runs M$.
Formally, $\equiv_{\indep}$ is defined as the transitive closure of the
relation $\equiv^1_{\indep}$, which in turn is defined as
$\sigma \equiv^1_{\indep} \sigma'$ iff
there is $\sigma_1, \sigma_2 \in T^*$ such that
$\sigma = \sigma_1.t.t'.\sigma_2$, 
$\sigma' = \sigma_1.t'.t.\sigma_2$,
and $t \indep t'$.
From the properties of $\indep$,
one can immediately see that $\equiv_{\indep}$ refines $\equiv$,
\ie, if $\sigma \equiv_{\indep} \sigma'$, then $\sigma \equiv \sigma'$.

Given a run $\sigma \in \runs M$, the
equivalence class of $\equiv_{\indep}$ to which $\sigma$ belongs
is called the \emph{Mazurkiewicz~trace} of~$\sigma$~\cite{Maz87},
denoted by $\trace{\indep,\sigma}$.
Each trace $\trace{\indep,\sigma}$
can equivalently be seen as
a labelled partial order $\dgraph{\indep,\sigma}$, traditionally
called the \emph{dependence graph} (see~\cite{Maz87} for a formalization),
satisfying that
a run belongs to the trace iff it is a linearization
of~$\dgraph{\indep,\sigma}$.

Sleep sets~\cite{God96} are another method for state-space reduction.
Unlike selective exploration, they
prune successors by looking at the past of the exploration, not the future.

\section{Parametric Partial Order Semantics}
\label{s:unfsem}

An unfolding is, conceptually, a tree-like structure of partial orders.
In this section, given an independence relation $\indep$ (our parameter) and a
system~$M$,
we define an unfolding semantics $\unf{M,\indep}$ with the following property:
each constituent partial order of $\unf{M,\indep}$ will correspond to
one dependence graph $\dgraph{\indep,\sigma}$, for some $\sigma \in \runs M$.
For the rest of this paper,
let $\indep$ be an arbitrary unconditional independence relation on~$M$.
We use prime event structures~\cite{NPW81},
a non-sequential, event-based model of concurrency, to define the
unfolding $\unf{M,\indep}$ of~$M$.

\begin{defn}[LES]
Given a set $A$, an \emph{$A$-labelled event structure}
($A$-LES, or LES in short)
is a tuple $\les \eqdef \tup{E, <, \cfl, h}$
where $E$ is a set of \emph{events},
${<} \subseteq E \times E$ is a strict partial order on $E$,
called \emph{causality relation},
$h \colon E \to A$ labels every event with an element of $A$, and
${\cfl} \subseteq E \times E$ is the
symmetric, irreflexive \emph{conflict relation}, satisfying
\begin{itemize}
\item
  for all $e \in E$, $\set{e' \in E \colon e' < e}$ is finite, and
  \eqtag{e:les1}
\item
  for all $e,e',e'' \in E$, if $e \cfl e'$ and $e' < e''$, then $e \cfl e''$.
  \eqtag{e:les2}
\end{itemize}
\end{defn}

The \emph{causes} of an event $e \in E$ are the set
$\causes e \eqdef \set{e' \in E \colon e' < e}$
of events that need to happen before~$e$ for~$e$ to happen.
A \emph{configuration} of $\mathcal{E}$ is any finite set~$C \subseteq E$
satisfying:
\begin{itemize}
\item
  (causally closed)
  for all $e \in C$ we have $\causes e \subseteq C$;
  \eqtag{c:con1}
\item
  (conflict free)
  for all $e, e' \in C$, it holds that $\lnot e \cfl e'$.
  \eqtag{c:con2}
\end{itemize}
Intuitively, configurations represent partially-ordered executions.
In particular, the \emph{local configuration} of $e$ is the
$\subseteq$-minimal configuration that contains $e$, \ie
$[e] \eqdef \causes e \cup \set e$.
We denote by $\conf \les$ the set of configurations of $\les$.
Two events $e, e'$ are in \emph{immediate conflict}, $e \icfl{} e'$,
iff $e \cfl e'$ and both $\causes e \cup [e']$ and $[e] \cup \causes{e'}$
are configurations.
%
Lastly,
given two LESs
$\les \eqdef \tup{E, <, \cfl, h}$ and
$\les' \eqdef \tup{E', <', \cfl', h'}$, we say that
$\les$ is a \emph{prefix} of~$\les'$, written $\les \ispref \les'$,
when
$E \subseteq E'$,
$<$ and $\cfl$ are the projections of $<'$ and $\cfl'$ to~$E$,
and $E \supseteq \set{e' \in E' \colon e' < e \land e \in E}$.

Our semantics will unroll the system~$M$ into a LES $\unf{M,\indep}$
whose events are labelled by transitions of~$M$.
Each configuration of $\unf{M,\indep}$ will correspond to the
dependence graph $\dgraph{\indep,\sigma}$ of some $\sigma \in \runs M$.
For a LES $\tup{E, <, \cfl, h}$, we define the
\emph{interleavings} of $C$ as
$\inter C \eqdef
\set{h(e_1), \ldots, h(e_n) \colon e_i, e_j \in C \land e_i < e_j \implies i < j}$.
Although for arbitrary LES
$\inter C$ may contain sequences not in $\runs M$,
the definition of $\unf{M,\indep}$ will ensure
that $\inter C \subseteq \runs M$.
Additionally, since all sequences in $\inter C$ belong to the same trace,
all of them reach the same state.
Abusing the notation, we define $\state C \eqdef \state \sigma$
if $\sigma \in \inter C$.
The definition is neither well-given nor unique for arbitrary
LES, but will be so for the unfolding.

We now define $\unf{M,\indep}$.
Each event will be inductively identified by a canonical name
of the form $e \eqdef \tup{t,H}$, where $t \in T$ is a transition of~$M$
and $H$ a configuration of $\unf{M,\indep}$.
Intuitively, $e$ represents the occurrence of $t$
after the \emph{history} (or the causes) $H \eqdef \causes e$.
%
The definition will be inductive.
The base case inserts into the unfolding a special \emph{bottom event}
$\bot$ on which every event causally depends.
The inductive case iteratively extends the unfolding with one event.
We define
the set $\mathcal{H}_{\les,\indep,t}$ of candidate \emph{histories} for a
transition~$t$ in an LES $\les$ as the set which
contains exactly all configurations~$H$ of~$\les$ such that
\begin{itemize}
\item
  transition $t$ is enabled at $\state H$, and 
\item
  either $H = \set \bot$ or all $<$-maximal events $e$ in~$H$
  satisfy that $h(e) \depen t$,
\end{itemize}
where~$h$ is the labelling function in~$\les$.
Once an event $e$ has been inserted into the unfolding,
its associated transition $h(e)$ may be dependent with $h(e')$ for
some~$e'$ already present and outside the history of~$e$.
Since the order of occurrence of $e$ and $e'$ matters, we need to prevent
their occurrence within the same configuration,
as configurations represent equivalent executions.
So we introduce a conflict between~$e$ and~$e'$.
The set $\mathcal{K}_{\les,\indep,e}$
of \emph{events conflicting} with $e \eqdef \tup{t,H}$
thus contains any event $e'$ in $\les$ with
$e' \notin [e]$ and $e \notin [e']$ and $t \depen h(e')$.

Following common practice~\cite{EH08},
the definition of $\unf{M,\indep}$ proceeds in two steps.
We first define (\cref{d:finite.prefix}) the collection of all prefixes of the unfolding.
Then we show that there exists only one \mbox{$\ispref$-maximal} element in
the collection, and define it to be \emph{the} unfolding (\cref{d:unf}).

\begin{defn}[Finite unfolding prefixes]
\label{d:finite.prefix}
The set of \emph{finite unfolding prefixes} of $M$ under the independence
relation $\indep$
is the smallest set of LESs that satisfies the following conditions:
\begin{enumerate}
\item
  The LES having exactly one event $\bot$, empty causality and conflict
  relations, and $h (\bot) \eqdef \varepsilon$
  is an unfolding prefix.
\item
  Let~$\les$ be an unfolding prefix containing
  a history~$H \in \mathcal{H}_{\les,\indep,t}$ for some transition~$t \in T$.
  Then, the LES $\tup{E, <, {\cfl}, h}$
  resulting from extending~$\les$ with a new event~$e \eqdef \tup{t, H}$
  and satisfying the
  following constraints is also an unfolding prefix of $M$:
  \begin{itemize}
  \item for all $e' \in H$, we have $e' < e$;
  \item
    for all $e' \in \mathcal{K}_{\les,\indep,e}$, we have $e \cfl e'$; and
    $h (e) \eqdef t$.
  \end{itemize}
\end{enumerate}
\end{defn}

Intuitively, each unfolding prefix contains the dependence graph
(configuration) of one or more executions of~$M$ (of finite length).
The unfolding starts from $\bot$, the ``root'' of the tree,
and then iteratively adds events enabled by some configuration
until saturation, \ie, when no more events can be added.
Observe that the number of unfolding prefixes as per \cref{d:finite.prefix}
will be finite iff all runs of~$M$ terminate.
Due to lack of space, we give the definition of \emph{infinite} unfolding
prefix in \cref{x:unfsem},
as the main ideas of this section are well conveyed using only finite
prefixes.
In the sequel, by \emph{unfolding prefix} we mean a finite or infinite one.

Our first task is checking that each unfolding prefix is indeed a LES
(\cref{l:pref.les}).
Next one shows that the configurations of every unfolding prefix
correspond the Mazurkiewicz traces of the system, \ie,
for any configuration $C$,
$\inter C = \trace{\indep,\sigma}$ for some $\sigma \in \runs M$
(\cref{l:soundun}).
This implies that the definition of $\inter C$ and $\state C$ is well-given
when $C$ belongs to an unfolding prefix.
The second task is defining \emph{the} unfolding $\unf{M,\indep}$ of~$M$.
Here, we prove that the set of unfolding prefixes equipped with
relation~$\ispref$ forms a complete \mbox{join-semilattice}
(\cref{l:complete.lattice}).
This implies the existence of a unique
\mbox{$\ispref$-maximal} element:

\begin{defn}[Unfolding]
\label{d:unf}
The \emph{unfolding} $\unf{M,\indep}$ of~$M$ under the independence relation
$\indep$ is the \emph{unique} \mbox{$\ispref$-maximal}
element in the set of unfolding prefixes of~$M$ under~$\indep$.
\end{defn}

Finally we verify that the definition is well given
and that the unfolding is \emph{complete},
\ie, every run of the system is
represented by a unique configuration of the unfolding.

\begin{restatable}{theo}{completeun}
\label{r:unf.complete}
The unfolding $\unf{M,\indep}$ exists and is unique.
Furthermore, for any non-empty run $\sigma$ of~$M$,
there exists a unique configuration~$C$ of
$\unf{M,\indep}$ such that $\sigma \in \inter C$.
\end{restatable}

\SetKwFunction{explore}{Explore}

\newcommand\node[3]{%
{\scriptsize%
\if\relax\detokenize{#1}\relax
$\emptyset$\else #1\fi
\hspace{1pt}|\hspace{1pt}
\if\relax\detokenize{#2}\relax
$\emptyset$\else #2\fi
\hspace{1pt}|\hspace{1pt}
\if\relax\detokenize{#3}\relax
$\emptyset$\else #3\fi}%
}

\begin{figure}[t]
\centering
\includegraphics[scale=1.00]{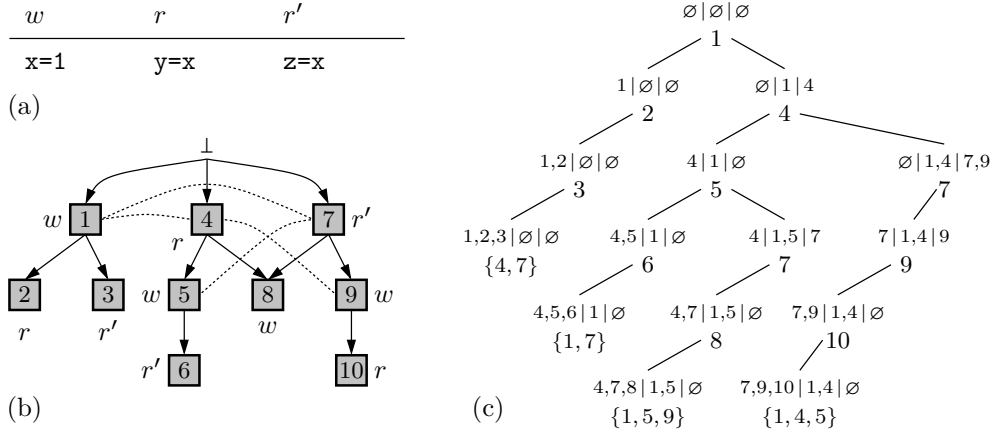}
\caption{Running example.
(a) A concurrent program; (b) its unfolding semantics.
(c) The exploration performed by \cref{a:a1},
where each node \node{C}{D}{A} represents
one call to the function \explore{$C,D,A$}.
The set $X$ underneath each leaf node
is such that the value of variable $U$ in \cref{a:a1} at the leaf is
$U = C \cup D \cup X$.
At \node{}{}{}, the alternative taken is $\set{4}$, and
at \node{4}{1}{} it is $\set{7}$.}
\label{f:running}
\end{figure}

\begin{exampl}[Programs]
\Cref{f:running}~(a) shows a concurrent program,
where process $w$ writes global variable and processes~$r$ and~$r'$ read it.
We can associate various semantics to it.
Under an empty independence relation, the unfolding would be the
computation tree, where executions would be totally ordered.
Considering (the unique transition of) $r$ and $r'$ independent,
and $w$ dependent on them, we get the unfolding
shown in~\cref{f:running}~(b).

Events are numbered from~1 to~10, and labelled with a transition.
Arrows represent causality between events and dotted lines immediate conflict.
The Mazurkiewicz trace of each deadlocking
execution is represented by a unique \mbox{$\subseteq$-maximal}
configuration, \eg, the run $w.r.r'$ yields
configuration $\set{1,2,3}$,
where the two possible interleavings reach the same state.
The canonic name of, \eg, event~1 is $\tup{w, \set \bot}$.
For event~2 it is
$\tup{r, \set{\bot, 1}}$.
Let $\ppref$ be the unfolding prefix that contains events $\set{\bot,1,2}$.
\Cref{d:finite.prefix} can extend it with three possible events: 3,~4, and~7.
Consider transition~$r'$. Three configurations of $\ppref$ enable $r'$:
$\set{\bot}, \set{\bot,1}$ and~$\set{\bot,1,2}$.
But since $\lnot (h(2) \depen r')$, 
only the first two will be in $\mathcal{H}_{\ppref,\indep,r'}$,
resulting in events
$3 \eqdef \tup{r', \set{\bot,1}}$ and
$7 \eqdef \tup{r', \set \bot}$.
Also, $\mathcal{K}_{\ppref,\indep,7}$ is $\set{1}$,
as $w \depen r'$.
The~4 maximal configurations are
$\set{1,2,3}$,
$\set{4,5,6}$,
$\set{4,7,8}$ and
$\set{7,9,10}$,
\resp reaching the states
$\tup{x,y,z}$ =
$\tup{1,1,1}$,
$\tup{1,0,1}$,
$\tup{1,0,0}$ and
$\tup{1,1,0}$,
assuming that variables start at~0.
\end{exampl}

\begin{exampl}[Comparison to Petri Net Unfoldings]
In contrast to our parametric semantics,
classical unfoldings of Petri nets~\cite{ERV02} 
use a fixed independence relation, specifically
the complement of the following one (valid only for safe nets):
given two transitions $t$ and $t'$,
\[
t \mathrel{\depen_n} t'
\text{ iff }
(\post t \cap \pre{t'} \ne \emptyset) \text{~ or ~}
(\post{t'} \cap \pre{t} \ne \emptyset) \text{~ or ~}
(\pre{t'} \cap \pre{t} \ne \emptyset),
\]
where $\pre t$ and $\post t$ are respectively the \emph{preset} and \emph{postset}
of~$t$.
Classic Petri net unfoldings (of safe nets)
are therefore a specific instantiation of our semantics.
A well known limitation of classic unfoldings
are transitions that ``\textit{read}'' places, \eg,~$t_1$ and~$t_2$
in~\cref{f:2read}~(a).
Since $t_1 \depen_n t_2$, the
classic unfolding, \cref{f:2read}~(b), sequentializes all their occurrences.
A solution to this is the so-called \emph{place replication (PR)
unfolding}~\cite{MR95}, or alternatively \emph{contextual unfoldings} (which
anyway internally are of asymptotically the same size as the PR-unfolding).

This problem vanishes with our parametric unfolding.
It suffices to use a dependency relation
${\depen'_n} \subset {\depen_n}$ that makes transitions that
``\textit{read}'' common places independent.
The result is that our unfolding, \cref{f:2read}~(c),
can be of the same size as the PR-unfolding, \ie,
exponentially more compact than the classic unfolding.
For instance, when \cref{f:2read}~(a) is generalized to~$n$ \emph{reading}
transitions, the classic unfolding would
have~$\bigo{n!}$ copies of~$t_3$, while ours would have~$\bigo{2^n}$.
The point here is that our semantics naturally accommodate a more suitable notion
of independence without resorting to specific ad-hoc tricks.

Furthermore, although this work is restricted to \emph{unconditional} independence,
we conjecture that an adequately
restricted \emph{conditional} dependence would suffice, \eg, the one of~\cite{KP92}.
Gains achieved in such setting would be difficult with classic
unfoldings. 
\end{exampl}

\begin{figure}[t]
\begin{center}
\includegraphics[scale=1.00]{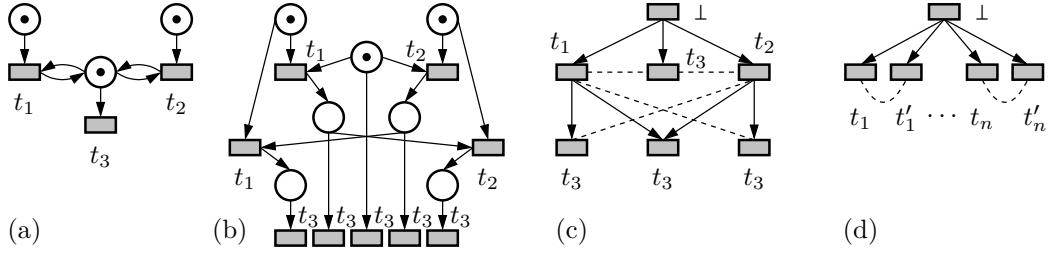}
\end{center}
\caption{(a) A Petri net; (b) its classic unfolding; (c) our parametric semantics.}
\label{f:2read}
\end{figure}

\section{Stateless Unfolding Exploration Algorithm}
\label{s:unfalgo}

We present a 
DPOR algorithm to explore an arbitrary event
structure (\eg, the one of~\cref{s:unfsem}) instead of sequential executions.
Our 
algorithm explores one configuration at a time and organizes the exploration
into a binary tree. \Cref{f:running}~(c) shows an example.
The algorithm is 
optimal~\cite{AAJS14}, in the sense that
no configuration is ever visited twice in the tree.

For the rest of the paper, let
$\unf{\indep,M} \eqdef \tup{E, <, {\cfl}, h}$ be the unfolding of~$M$
under~$\indep$, which we abbreviate as~$\uunf$.
For this section we assume that $\uunf$ is finite, \ie,
that all computations of~$M$ terminate.
This is only to ease presentation, we relax this assumption in~\cref{s:cutoffs}.

We give some new definitions.
Let~$C$ be a configuration of~$\uunf$.
The \emph{extensions} of $C$, written $\ex C$, are all those
events outside~$C$ whose causes are included in~$C$. Formally,
$
\ex C \eqdef
\set{e \in E \colon e \notin C \land \causes e \subseteq C}
$.
We let $\en C$ denote the set of events \emph{enabled} by~$C$, \ie,
those corresponding to the transitions enabled at~$\state C$,
formally defined as
$
\en C \eqdef
\set{e \in \ex C \colon C \cup \set e \in \conf{\uunf}}
$.
All those events in $\ex C$ which are
not in~$\en C$ are the
\emph{conflicting extensions},
$
\cex C \eqdef
\set{e \in \ex C \colon \exists e' \in C,\, e \icfl{} e'}
$.
Clearly, sets~$\en C$ and~$\cex C$ partition the set~$\ex C$.
Lastly, we define
$\ficfl e \eqdef \set{e' \in E \colon e \icfl{} e'}$,
and
$\ficfl[U] e \eqdef \ficfl e \cap U$.
The difference between both is that
$\ficfl e$ contains events from \emph{anywhere} in the unfolding structure,
while $\ficfl[U] e$ can only \emph{see} events in~$U$.


\SetKwProg{Proc}{Procedure}{}{}
\SetKwIF{If}{ElseIf}{Else}{if}{}{else if}{else}{end}
\SetKwFor{ForEach}{foreach}{}{end}
\SetKwFor{For}{for}{}{end}
\SetInd{0.1em}{0.9em}
\SetNlSty{}{\color{gray}}{}
\SetVlineSkip{0em}

\SetKwFunction{extend}{Extend}
\SetKwFunction{extendcut}{Extend'}
\SetKwFunction{remove}{Remove}
\SetKwFunction{ena}{en}
\SetKwFunction{alternatives}{Alt}
\SetKwFunction{k}{k} 

The algorithm is given in \cref{a:a1}.
\explore{$C,D,A$}, the main procedure, is given
the configuration that is to be explored as the parameter $C$.
The parameter~$D$ (for \emph{disabled}) is the set of set of events
that have already been explored and prevents that \explore{} repeats work.
It can be seen as a \textit{sleep set}~\cite{God96}.
Set~$A$ (for \emph{add}) is occasionally used to guide the
direction of the exploration.

Additionally,
a global set~$U$ stores all events presently known
to the algorithm. Whenever some event can safely be discarded from memory,
\remove will
move it from~$U$ to~$G$ (for \emph{garbage}). Once in~$G$, it can be
discarded at any time, or be preserved in $G$ in order to save
work when it is re-inserted in~$U$. Set~$G$ is thus
our \emph{cache memory} of events.

\begin{algorithm}[t]
\DontPrintSemicolon
\setstretch{1.1}

Initially, set $U \eqdef \set \bot$, set $G \eqdef \emptyset$,
and call \explore{$\set \bot$, $\emptyset$, $\emptyset$}. \;

\BlankLine

\begin{multicols}{2}
\Proc{\explore{$C, D, A$}}{
\extend{$C$} \;
\lIf{$\ena{$C$} = \emptyset$}{
   \KwRet
}
\label{l:a1ret}

\If{$A = \emptyset$}{
   Choose $e$ from \ena{$C$} \;
}
\Else{
   Choose $e$ from $A \cap \ena{$C$}$ \;
   \label{l:a1choose}
}
\explore{$C \cup \set e, D, A \setminus \set e$} \;
\label{l:a1l}
\If{$\exists J \in \alternatives{$C, D \cup \set e$}$}{
   \explore{$C, D \cup \set e, J \setminus C$}
   \label{l:a1r}
}
\remove{$e, C, D$}
}

\Proc{\extend{$C$}}{
Add $\ex C$ to $U$
}

\BlankLine
\BlankLine
\Proc{\remove{$e, C, D$}}{
Move $\set e \setminus Q_{C,D,U}$ from $U$ to~$G$ \;
\ForEach{$\hat e \in \ficfl[U] e$}{
   Move $[\hat e] \setminus Q_{C,D,U}$ from~$U$ to~$G$ \;
}
}
\end{multicols}
\medskip

\caption{An unfolding-based POR exploration algorithm.}
\label{a:a1}
\end{algorithm}

The key intuition in \cref{a:a1} is as follows.
A call to \explore{$C,D,A$} visits all maximal configurations of~$\uunf$
which contain~$C$ and do not contain~$D$; and the first one explored will
contain~$C \cup A$.
\Cref{f:running}~(c) gives one execution,
tree nodes are of the form \node{C}{D}{A}.

The algorithm first updates~$U$ with all extensions of~$C$ (procedure \extend).
If~$C$ is a maximal configuration, then there is nothing to do, it backtracks.
If not, it chooses an event in~$U$ enabled at~$C$,
using the function~$\ena{$C$} \eqdef \en C \cap U$.
If~$A$ is empty, any enabled event can be taken.
If not, $A$ needs to be explored and $e$ must come from the intersection.
Next it makes a recursive call (left subtree), where it explores
\emph{all} configurations containing all events in $C \cup \set e$ and
no event from~$D$.
Since \explore{$C,D,A$} had to visit all maximal configurations containing $C$,
it remains to visit those containing~$C$ but not~$e$,
but only if there exists at least one!
Thus, we determine whether~$\uunf$ has a maximal configuration that
contains~$C$,
does not contain~$D$ and
does not contain~$e$.
Function \alternatives will return a set of events that witness the existence of
such configuration (iff one exists).
If one exists, we make a second recursive call (right subtree).
Formally, we call such witness an \emph{alternative}:

\begin{defn}[Alternatives]
\label{d:alt}
Given a set of events~$U \subseteq E$,
a configuration $C \subseteq U$,
and a set of events $D \subseteq U$,
an \emph{alternative} to $D$ after $C$ is any configuration
$J \subseteq U$ satisfying that
\begin{itemize}
\item
  $C \cup J$ is a configuration
  \eqtag{e:alt2}
\item
  for all events $e \in D$, there is
  some $e' \in C \cup J$ such that $e' \in \ficfl[U] e$.
  \eqtag{e:alt1}
\end{itemize}
\end{defn}

Function \alternatives{$X, Y$} returns all alternatives (in~$U$)
to~$Y$ after~$X$.
Notice that it is called as \alternatives{$C, D \cup \set e$}
from~\cref{a:a1}. Any returned alternative~$J$ witnesses
the existence of a maximal configuration $C'$
(constructed by arbitrarily extending $C \cup J$)
where $C' \cap (D \cup \set e) = \emptyset$.

Although \alternatives reasons about maximal configurations of $\uunf$, thus
potentially about events which have not yet been seen,
it can only look at events in~$U$.
So the set~$U$ needs to be large enough to contain enough \textit{conflicting
events} to satisfy~\cref{e:alt1}.
Perhaps surprisingly, it suffices to store only
events seen (during the past exploration) in immediate conflict with~$C$ and $D$.
Consequently, when the algorithm calls \remove, to clean from~$U$ events that are no longer
necessary (\ie, necessary to find alternatives in the future),
it needs to preserve at least those conflicting events.
Specifically, \remove will preserve in~$U$ the following events:
\[
Q_{C,D,U} \eqdef
C \cup D \cup \bigcup_{e \in C \cup D, e' \in \ficfl[U] e} [e'].
\]
That is, events in~$C$, in~$D$ and events in conflict with those.
An alternative definition that makes $Q_{C,D,U}$ smaller
would mean that \remove discards more events, which could
prevent a future call to \alternatives from discovering a maximal configuration
that needs to be explored.

We focus now on the correctness of \cref{a:a1}.
Every call to \explore{$C,D,A$}
explores a tree,
where the recursive calls at lines \cref{l:a1l} and \cref{l:a1r} respectively
explore the left and right subtrees
(proof in \cref{c:btree}).
Tree nodes are tuples~$\tup{C,D,A}$ corresponding to the arguments of calls
to~\explore, \cf~\cref{f:running}.
We refer to this object as the \emph{call tree}.
For every node, both~$C$ and~$C \cup A$ are configurations, and~$D \subseteq \ex C$
(\cref{l:general}).
As the algorithm goes down in the tree
it monotonically increases the size of either~$C$ or~$D$.
Since~$\uunf$ is finite, this implies that the algorithm terminates:

\begin{restatable}[Termination]{theo}{thmtermination}
   \label{t:a1.termination}
   Regardless of its input, \cref{a:a1} always stops.
\end{restatable}

Next we check that \cref{a:a1} never visits twice the same
configuration,
which is why it is called an \emph{optimal}~POR~\cite{AAJS14}.
We show that for every node in the call tree, the set of configurations
in the left and right subtrees are disjoint (\cref{l:c3c4}).
This implies that:

\begin{restatable}[Optimality]{theo}{thmoptimality}
   \label{t:a1.optimality}
   Let $\tilde C$ be a maximal configuration of~$\uunf$.
   Then \explore{$\cdot,\cdot,\cdot$} is called \emph{at most once}
   with its first parameter being equal to~$\tilde C$.
\end{restatable}

Parameter~$A$ of \explore plays a central role in making \cref{a:a1} optimal.
It is necessary to ensure that, once the algorithm decides to explore some
alternative~$J$, such an alternative is visited first. Not doing so makes it
possible to extend~$C$ in such a way that no maximal configuration
can ever avoid including events in~$D$. Such a configuration,
referred as a \textit{sleep-set blocked} execution in~\cite{AAJS14},
has already been explored before.

Finally, we ensure that \cref{a:a1} visits every maximal configuration
of~$\uunf$.
This essentially reduces to showing that it
makes the second recursive call, \cref{l:a1r},
whenever there exists some unexplored maximal configuration 
not containing $D \cup \set e$.
The difficulty of proving so (\cref{l:compl.step})
comes from the fact that \cref{a:a1} \emph{only} sees events in~$U$.
Due to space constraints,
we omit an additional result on the memory
consumption, \cf \cref{x:memory.lazy}.

\begin{restatable}[Completeness]{theo}{thmcompleteness}
   \label{t:a1.completeness}
   Let $\tilde C$ be a maximal configuration of~$\uunf$.
   Then \explore{$\cdot,\cdot,\cdot$} is called \emph{at least once}
   with its first parameter being equal to~$\tilde C$.
\end{restatable}


\section{Improvements}
\label{s:impro}

\subsection{State Caching}

Stateless model checking algorithms explore only one configuration of $\uunf$ at
a time, thus potentially under-using remaining available memory.
A desirable property for an algorithm is the capacity to exploit
all available memory without imposing the liability of actually requiring it.
The algorithm 
in \cref{s:unfalgo} satisfies this property. The set~$G$,
storing events discarded from~$U$, can be cleaned at discretion, \eg, when the
memory is approaching full utilisation. Events cached in~$G$ are exploited in
two different ways.

First,
whenever an event in~$G$ shall be included again in~$U$,
we do not need to reconstruct it in memory (causality, conflicts, \etc).
In extreme cases, this might happen frequently.
Second,
using the result of the next section,
cached events help prune the number of maximal configurations to visit.
This means that our POR potentially visits \emph{fewer} final states
than the number of configurations of~$\uunf$,
thus conforming to the requirements of a \emph{super-optimal DPOR}.
The larger~$G$ is, the fewer configurations will be explored.

\subsection{Non-Acyclic State Spaces}
\label{s:cutoffs}

In this section we remove the assumption that $\unf{M,\indep}$ is finite.
We employ the notion of cutoff events~\cite{Mcm93}.
While cutoffs are a standard tool for unfolding pruning, their application
to our framework brings unexpected problems.

The core question here is preventing \cref{a:a1} from getting stuck in the
exploration of an infinite configuration.
We need to create the illusion that maximal configurations are finite.
We achieve this by substituting procedure \extend in \cref{a:a1} with another
procedure \extendcut that operates as \extend except that it only
adds to~$U$ an event from $e \in \ex C$ if the predicate $\iscutoff{e,U,G}$
evaluates to false.
We define~$\iscutoff{e,U,G}$ to hold iff
there exists some event $e' \in U \cup G$ such that
\begin{equation}
\label{e:cutoff}
\state{[e]} = \state{[e']}
\text{ ~ and ~ }
|[e']| < |[e]|.
\end{equation}
We refer to $e'$ as the \emph{corresponding} event of~$e$, when it exists.
This definition declares $e$ cutoff as function of~$U$ and~$G$.
This has important consequences.
An event $e$ could be declared cutoff while exploring one maximal configuration and
non-cutoff while exploring the next, as the corresponding
event might have disappeared from $U \cup G$.
This is in stark contrast to the classic unfolding construction, where events are
declared cutoffs \emph{once and for all}.
The main implication is that the standard
argument~\cite{Mcm93,ERV02,BHKTV14} invented by McMillan for proving
completeness fails.
We resort to a completely different argument for proving
completeness of our algorithm (see~\cref{x:cutoffs}),
which we are forced to skip in view of the lack of space.

We focus now on the correction of \cref{a:a1} using \extendcut instead of
\extend.
A \emph{causal cutoff} is any event~$e$ for which there is some $e' \in [e]$
satisfying \cref{e:cutoff}.
It is well known that causal cutoffs define a finite prefix of~$\uunf$ as per
the classic saturation definition~\cite{BHKTV14}.
Also, $\iscutoff{e,U,G}$ always holds for causal cutoffs,
regardless of the contents of~$U$ and~$G$.
This means that the modified algorithm can only explore configurations
from a finite prefix. It thus necessarily terminates.
As for optimality, it is unaffected by the use of cutoffs, existing proofs for
\cref{a:a1} still work.
Finally, for completeness we prove the following result,
stating that local reachability (\eg, fireability of transitions of~$M$) is
preserved:

\begin{restatable}[Completeness]{theo}{thmcutoffcompleteness}
   \label{t:cutoff.completeness}
   For any reachable state $s \in \reach M$,
   \cref{a:a1} updated with the cutoff mechanism described above explores one
   configuration~$C$ such that for some $C' \subseteq C$ it holds that
   $\state{C'} = s$.
\end{restatable}

Lastly, we note that this cutoff approach imposes no liability on what events
shall be kept in the prefix, set~$G$ can be cleaned at discretion.
Also, redefining~\cref{e:cutoff} to use adequate orders~\cite{ERV02} is
straightforward, \cf~\cref{x:cutoffs} (in our proofs we actually assume adequate orders).

\section{Experiments}
\label{s:experiments}

As a proof of concept,
we implemented our algorithm in a new
explicit-state model checker baptized \poet (Partial Order Exploration
Tool).\footnote{Source code and benchmarks
available from:
\url{http://www.cs.ox.ac.uk/people/marcelo.sousa/poet/}.}
Written in Haskell, a lazy functional language,
it analyzes programs from a restricted fragment of the C language and supports POSIX threads.
The analyzer accepts deterministic programs, implements a variant of \cref{a:a1} where
the computation of the alternatives is memoized,
and supports cutoffs events with the criteria defined in \cref{s:impro}.

We ran \poet on a number of multi-threaded C programs.
Most of them are adapted from benchmarks of the Software Verification
Competition~\cite{URL:SVCOMP15};
others are used in related works~\cite{GFYS07,YCGK08,AAJS14}.
We investigate the characteristics of average program unfoldings (depth,
width, \etc) as well as the frequency and impact of cutoffs on the exploration.
We also compare \poet with \nidhugg~\cite{AAJLS15}, a state-of-the-art 
stateless model checking for multi-threaded C programs that implements
Source-DPOR~\cite{AAJS14}, an efficient but non-optimal DPOR.
All experiments were run on an Intel Xeon CPU with 2.4\,GHz and 4\,GB memory.
\Cref{table:resultsacylic,table:resultscyclic} give our experimental data for
programs with acyclic and non-acyclic state spaces,
respectively.

\newcommand\newrow{\\[-4pt]}
\newcommand\param[1]{\footnotesize(#1)}

\begin{table}[!t]

\caption{Programs with acyclic state space.
Columns are:
$|P|$:              \nr of threads;
$|I|$:              \nr of explored traces;
$|B|$:              \nr of sleep-set blocked executions;
$t(s)$:             running time;
$|E|$:              \nr of events in $\uunf$;
$|E_{\text{cut}}|$: \nr of cutoff events;
$|\Omega|$:         \nr of maximal configurations;
$\tup{|U_\Omega|}$: \avg \nr of events in~$U$ when exploring a maximal configuration.
A $*$ marks programs containing bugs.
{\tt <7K} reads as ``\emph{fewer than 7000}''.
}
\label{table:resultsacylic}

\setlength\tabcolsep{3.3pt}
\def\sep{\hspace{10pt}}
\def\tinysep{\hspace{4pt}}

\centering
\tt
\begin{tabular}{lr@{\sep}rrr@{\sep}rrrr@{\sep}rrrrr}
\toprule
  \multicolumn{2}{l}{\rm \normalsize Benchmark}
& \multicolumn{3}{l}{\rm \normalsize \nidhugg}
& \multicolumn{4}{l}{\rm \normalsize \poet (without cutoffs)}
& \multicolumn{5}{l}{\rm \normalsize \poet (with cutoffs)}
\\[-2pt]
  \cmidrule(r){1-2}
  \cmidrule(r){3-5}
  \cmidrule(r){6-9}
  \cmidrule(r){10-14}
  \rm Name 
& $|P|$
& $|I|$
& $|B|$
& $t(s)$
& $|E|$
& $|\Omega|$
& $\tup{|U_\Omega|}$
& $t(s)$
& $|E|$
& $|E_{\text{cut}}|$
& $|\Omega|$
& $\tup{|U_\Omega|}$
& $t(s)$
\\[-2pt]
\midrule
\rm\sc Stf            & 3  & 6    & 0   & 0.06 & 121    & 6    & 79     & 0.04 & 121   & 0   & 6    & 79    & 0.06 \newrow
\rm\sc Stf$*$         & 3  & --   & --  & 0.05 & --     & --   & --     & 0.02 & --    & --  & --   & --    & 0.03 \newrow
\rm\sc Spin08         & 3  & 84   & 0   & 0.08 & 2974   & 84   & 1506   & 2.04 & 2974  & 0   & 84   & 1506  & 2.93 \newrow
\rm\sc Fib            & 3  & 8953 & 0   & 3.36 & <185K  & 8953 & 92878  & 305  & <185K & 0   & 8953 & 92878 & 704 \newrow
\rm\sc Fib$*$         & 3  & --   & --  & 0.74 & --     & --   & --     & 81.0 & --    & --  & --   & --    & 133 \newrow
\rm\sc Ccnf\param 9   & 9  & 16   & 0   & 0.05 & 49     & 16   & 46     & 0.07 & 49    & 0   & 16   & 46    & 0.06 \newrow
\rm\sc Ccnf\param{17} & 17 & 256  & 0   & 0.15 & 97     & 256  & 94     & 5.76 & 97    & 0   & 256  & 94    & 6.09 \newrow
\rm\sc Ccnf\param{19} & 19 & 512  & 0   & 0.28 & 109    & 512  & 106    & 22.5 & 109   & 0   & 512  & 106   & 22.0 \newrow 
\rm\sc Ssb            & 5  & 4    & 2   & 0.05 & 48     & 4    & 38     & 0.03 & 46    & 1   & 4    & 37    & 0.03 \newrow
\rm\sc Ssb\param 1    & 5  & 22   & 14  & 0.06 & 245    & 23   & 143    & 0.11 & 237   & 4   & 23   & 140   & 0.11 \newrow
\rm\sc Ssb\param 3    & 5  & 169  & 67  & 0.12 & 2798   & 172  & 1410   & 3.51 & 1179  & 48  & 90   & 618   & 0.90 \newrow
\rm\sc Ssb\param 4    & 5  & 336  & 103 & 0.15 & <7K  & 340  & 3333   & 20.3 & 2179    & 74  & 142  & 1139  & 2.07 \newrow
\rm\sc Ssb\param 8    & 5  & 2014 & 327 & 0.85 & <67K & 2022 & 32782  & 4118 & <12K    & 240 & 470  & 6267  & 32.1
\\[-2pt]
\bottomrule
\end{tabular}
\end{table}

\begin{table}[!t]

\caption{Programs with non-terminating executions. Column $b$ is the loop
bound. The value is chosen based on experiments described in~\cite{AAJLS15}.}
\label{table:resultscyclic}

\setlength\tabcolsep{4pt}
\def\sep{\hspace{36pt}}

\centering
\tt
\begin{tabular}{lr@{\sep}rrrr@{\sep}rrrrr}
\toprule
  \multicolumn{2}{l}{\rm \normalsize Benchmark}
& \multicolumn{4}{l}{\rm \normalsize \nidhugg}
& \multicolumn{5}{l}{\rm \normalsize \poet (with cutoffs)}
\\[-2pt]
  \cmidrule(r){1-2}
  \cmidrule(r){3-6}
  \cmidrule(r){7-11}
  \rm Name 
& $|P|$
& $b$
& $|I|$
& $|B|$
& $t(s)$
& $|E|$
& $|E_{\text{cut}}|$
& $|\Omega|$
& $\tup{|U_\Omega|}$
& $t(s)$
\\[-2pt]
\midrule
\rm\sc {Szymanski}          & 3 & -- & 103    & 0  & 0.07  & 1121  & 313 & 159 & 591   & 0.36 \newrow
\rm\sc {Dekker}             & 3 & 10 & 199    & 0  & 0.11  & 217   & 14  & 21  & 116   & 0.07 \newrow
\rm\sc {Lamport}            & 3 & 10 & 32     & 0  & 0.06  & 375   & 28  & 30  & 208   & 0.12 \newrow
\rm\sc {Peterson}           & 3 & 10 & 266    & 0  & 0.11  & 175   & 15  & 20  & 100   & 0.05 \newrow
\rm\sc {Pgsql}              & 3 & 10 & 20     & 0  & 0.06  & 51    & 8   & 4   & 40    & 0.03 \newrow
\rm\sc {Rwlock}             & 5 & 10 & 2174   & 14 & 0.83  & <7317 & 531 & 770 & 3727  & 12.29 \newrow
\rm\sc {Rwlock\param{2}$*$} & 5 & 2  & --     & -- & 7.88  & --    & --  & --  & --    & 0.40 \newrow
\rm\sc {Prodcons}           & 4 & 5  & 756756 & 0  & 332.62 & 3111  & 568 & 386 & 1622  & 5.00 \newrow
\rm\sc {Prodcons\param{2}}  & 4 & 5  & 63504  & 0  & 38.49 & 640   & 25  & 15  & 374   & 1.61
\\[-2pt]
\bottomrule
\end{tabular}

\end{table}

For programs with acyclic state spaces (\cref{table:resultsacylic}), 
\poet with and without cutoffs seems to perform the same
exploration when the unfolding has no cutoffs, as expected.
Furthermore, the number of explored executions also coincides with \nidhugg 
when the latter reports 0 sleep-set blocked executions (\cf, \cref{s:unfalgo}),
providing experimental evidence of \poet's optimality.

The unfoldings of most programs in \cref{table:resultsacylic} do not contain
cutoffs.  All these programs are deterministic, and many of them
highly sequential (\textsc{Stf}, \textsc{Spin08}, \textsc{Fib}),
features known to make cutoffs unlikely.
\mbox{\textsc{Ccnf}($n$)} are concurrent programs composed of $n-1$ threads
where thread~$i$ and~$i+1$ race on writing one variable, and are independent of
all remaining threads.
Their unfoldings resemble \cref{f:2read}~(d), with $2^{(n-1)/2}$ traces
but only $\bigo n$ events.
Saturation-based unfolding methods would win here over both \nidhugg and \poet.

In the $\textsc{ssb}$ benchmarks, \nidhugg encounters
sleep-set blocked executions, thus performing sub-optimal exploration.
By contrast, \poet finds many cutoff events and
achieves a \emph{super-optimal} exploration, exploring fewer traces
than both \poet without cutoffs and \nidhugg.
The data shows that this \emph{super-optimality} results in substantial savings in runtime.

For non-acyclic state spaces (\cref{table:resultscyclic}), 
unfoldings are infinite. We thus
compare \poet with cutoffs and \nidhugg with a loop bound.
Hence, while \nidhugg performs bounded model checking,
\poet does complete verification.
The benchmarks include classical mutual exclusion protocols
($\textsc{Szymanski}, \textsc{Sekker}, \textsc{Lamport}$ and $\textsc{Peterson}$),
where \nidhugg is able to leverage an important static optimization that replaces each spin loop 
by a load and assume statement~\cite{AAJLS15}. Hence, the number of traces and maximal configurations is 
not comparable.
Yet \poet, which could also profit from this static optimization,
achieves a significantly better reduction thanks to cutoffs alone.
Cutoffs dynamically prune redundant unfolding branches and arguably constitute a more
robust approach than the load and assume syntactic substitution.
The substantial reduction in 
number of explored traces, several orders of magnitude in some cases,
translates in clear runtime improvements.
Finally, in our experiments, both tools were able to successfully discover assertion violations in 
\textsc{stf$*$}, \textsc{fib$*$} and \textsc{rwlock(2)$*$}.

In our experiments,
\poet's average maximal memory consumption (measured in events)
is roughly half of the size of the unfolding.
We also notice that most of these unfoldings are quite narrow and deep
($|E_{\text{cut}}| \div |E|$ is low) when
compared with standard benchmarks for Petri nets.
This suggests that they could be amenable for saturation-based unfolding
verification, possibly pointing the opportunity of applying these methods
in software verification.

\section{Related Work}
\label{s:related}

This work focuses on explicit-state POR,
as opposed to symbolic POR techniques exploited inside SAT
solvers,~\eg,~\cite{KWG09,GFYS07}.
Early POR statically computed the necessary transitions to fire at every
state~\cite{Val91,God96}.
Flanagan and Godefroid~\cite{FG05} first proposed to compute persistent sets
dynamically (DPOR).
However, even when combined with sleep~sets~\cite{God96}, DPOR was still unable
to explore exactly one interleaving per Mazurkiewicz trace.
Abdulla~et~\al~\cite{AAJS14,AAJLS15} recently proposed the first solution to this,
using a data structure called wakeup trees.
Their DPOR is thus optimal (ODPOR) in this sense.

Unlike us, ODPOR operates on an interleaved execution model.
Wakeup trees store chains of dependencies that assist the algorithm in
reversing races throughly.
Technically, each branch roughly correspond to one of our alternatives.
According to~\cite{AAJS14}, constructing and managing wakeup trees is
expensive.
This seems to be related with the fact that wakeup trees store canonical
linearizations of configurations, and need to canonize executions before
inserting them into the tree to avoid duplicates.
Such checks become simple linear-time verifications when seen as partial-orders.
Our alternatives are computed dynamically and exploit these partial orders,
although we do not have enough experimental data to compare with wakeup trees.
Finally, our algorithm is able to visit up to exponentially fewer
Mazurkiewicz traces (owing to cutoff events),
copes with non-terminating executions, and profits from
state-caching. The work in~\cite{AAJS14} has none of these features.

Combining DPOR with stateful search is challenging~\cite{YWY06}.
Given a state~$s$, DPOR relies on a complete exploration from~$s$ to determine
the necessary transitions to fire from~$s$,
but such exploration could be pruned if a state is revisited,
leading to unsoundness.
Combining both methods requires addressing this difficulty,
and two works did it~\cite{YWY06,YCGK08}, but for non-optimal DPOR.
By contrast, incorporating cutoff events into~\cref{a:a1} was straightforward.

Classic, saturation-based unfolding algorithms are also
related~\cite{Mcm93,ERV02,BHKTV14,KH14}.  They are inherently stateful,
cannot discard events from memory, but explore events instead of
configurations, thus may do exponentially less work.  They can
furthermore guarantee that the number of explored events will be at most the
number of reachable states, which at present seems a difficult goal for
PORs.  On the other hand, finding the events to extend the unfolding is
computationally harder.  In~\cite{KH14}, Kähkönen and Heljanko use
unfoldings for concolic testing of concurrent programs.  Unlike ours, their
unfolding is not a semantics of the program, but rather a means for
discovering all concurrent program paths.

While one goal of this paper is establishing an (optimal) POR exploiting
the same commutativity as some non-sequential semantics, a longer-term goal is
building formal connections between the latter and PORs.
Hansen and~Wang~\cite{HW11} presented a characterization of
(a class of) stubborn sets~\cite{Val91} in terms of configuration structures, another
non-sequential semantics more general than event structures.
We shall clarify that while we restrict ourselves to commutativity-based PORs,
they attempt a characterization of stubborn sets, which do not necessarily rely
on commutativity.

\section{Conclusions}
\label{s:concl}


In the context of commutativity-exploiting POR,
we introduced an optimal DPOR that leverages on cutoff events to prune
the number of explored Mazurkiewicz traces,
copes with non-terminating executions, and uses state caching to
speed up revisiting events.
The algorithm 
provides a new view to DPORs as
algorithms exploring an object with richer structure.
In future work, we plan exploit this richer structure to further reduce the
number of explored traces for both PORs and saturation-based unfoldings.

\bibliographystyle{plain}
\bibliography{refs}

\newpage
\appendix

\section{Proofs: Unfolding Semantics}
\label{x:unfsem}

In \cref{s:unfsem} we defined the set of finite unfolding prefixes of~$M$ under
the independence relation~$\indep$.
If~$M$ has only terminating executions, \ie, all elements in $\runs M$ are
finite,
then all unfolding prefixes are finite.
However, if it has non-terminating executions, then we need to also consider its
infinite unfolding prefixes.
We will achieve this in \cref{d:prefix}.
First we need a technical definition and some results about it.
Let
\[
F \eqdef \set{
\tup{E_1, <_1, \cfl_1, h_1},
\tup{E_2, <_2, \cfl_2, h_2},
\ldots}
\]
be a finite or infinite set of unfolding prefixes of~$M$ under~$\indep$.
We define the \emph{union} of all of them as the LES
$\union F \eqdef \tup{E, {<}, {\cfl}, h}$, where
\[
E \eqdef \bigcup_{1 \le i} E_i
\qquad
{<} \eqdef \bigcup_{1 \le i} {<_i}
\qquad
h \eqdef \bigcup_{1 \le i} h_i,
\]
and $\cfl$ is the $\subseteq$-minimal relation on $E \times E$ that
satisfies \cref{e:les2} and such that
$e \cfl e'$ holds
for any two events $e, e' \in E$ if
\begin{equation}
\label{e:union}
e \notin [e'] \text{ and }
e' \notin [e] \text{ and }
h(e) \depen h(e').
\end{equation}
Since every element of $F$ is a LES, clearly $\union F$ is also a LES,
\cref{e:les1} and \cref{e:les2} are trivially satisfied.
Notice that all events in $E_1, E_2, E_3, \ldots$ are pairs of the form
$\tup{t,H}$, and the union of two or more $E_i$'s will merge many
\textit{equal events}.  Indeed, two events
$e_1 \eqdef \tup{t_1, H_1}$ and
$e_2 \eqdef \tup{t_2, H_2}$
are equal iff $t_1 = t_2$ and $H_1 = H_2$.

\begin{defn}[Unfolding prefixes, finite or infinite]
\label{d:prefix}
The set of \emph{unfolding prefixes} of $M$ under the independence
relation~$\indep$ contains all
finite unfolding prefixes, as defined by \cref{d:finite.prefix},
together with those constructed by:
\begin{itemize}
\item
  For any infinite set~$X$ of 
  unfolding prefixes, $\union X$ is also an unfolding prefix.
\end{itemize}
\end{defn}

Our first task is verifying that each unfolding prefix is indeed a LES.
Conditions \cref{e:les1,e:les2} are satisfied by construction.
We verify the following:

\begin{lem}
   \label{l:pref.les}
   For any unfolding prefix $\ppref \eqdef \tup{E, {<}, {\cfl}, h}$ we have the
   following:
   \begin{enumerate}
   \item The relation $<$ is a strict partial order.
   \item The relation $\cfl$ is irreflexive.
   \end{enumerate}
\end{lem}

\begin{proof}
   Assume that $\ppref$ is finite. This means that it has been constructed with
   \cref{d:finite.prefix}.  We prove both statements by induction.
   
   \emph{Base case.}
   The prefix containing only $\bot$ trivially satisfies both statements.
   
   \emph{Step case.}
   We prove both statements separately.
     Clearly $e < e$ does not hold, as
     every event introduced by \cref{d:finite.prefix} is a causal successor of only
     events that were already present in the unfolding prefix.
     Furthermore, the insertion of an event does not change
     the causal relations existing in the preceding unfolding prefix.
     The relation $<$ is also transitive, as the history of a configuration is
     causally closed.

     As for the second statement, we prove it
     by contradiction.
     Assume that $e \cfl e$ and that $e$ has been inserted into $\ppref$
     by applying \cref{d:finite.prefix} to the
     prefix~$\ppref'$.  Clearly, $e \notin \mathcal{D}_{\ppref',e}$, so the
     conflict has not been inserted when extending $\ppref$ with $e$.
     It must be the case, then, that \cref{d:finite.prefix}
     has inserted another event $e'$ in $E$ after inserting $e$,
     and that $e' \in \causes e$ and $e' \cfl e$.
     This is also not possible since, by definition,
     when inserting $e'$ on a prefix $\ppref''$ no causal successor of $e'$ can
     be present in $\mathcal{D}_{\ppref'',e'}$.

   Assume now that $\ppref$ is not finite.
   Then it is the union of an infinite family of finite unfolding prefixes, each
   one of them satisfy the above.
   We prove again both statements separately.

     First statement.
     For any event $e \eqdef \tup{t,H} \in E$, necessarily $e < e$ cannot hold,
     as $e$ comes from some of the finite prefixes.
     Now, if $e$ belongs to several finite prefixes, by construction they agree
     on which events are causal predecessors of $e$.
     If the union contains a cycle
     \[
     e_1 < e_2 < \ldots < e_n < e_1,
     \]
     then all $n$ events are present in any finite prefix to which $e_n$
     belongs. As a result all of them are in $\causes{e_n}$, which is clearly
     impossible.

     Second statement.
     It cannot be the case that $e \cfl e$ in $\ppref$ but
     $\lnot (e \cfl e)$ in any finite prefix that gives rise to $\ppref$, by
     definition of $\union \cdot$.
     So since $\lnot (e \cfl e)$ holds for any finite prefix, then
     $\lnot (e \cfl e)$ holds for $\ppref$.
\end{proof}

We now need to prove some facts about $\union \cdot$.

\begin{lem}
\label{l:f.finite}
If $F$ is a finite set of unfolding prefixes constructed by
\cref{d:finite.prefix},
then $\union F$ is also a finite prefix constructed by
\cref{d:finite.prefix}.
\end{lem}

\begin{proof} (Sketch).
The proof proceeds by induction on the size $n$ of $F$.
If $n = 1$ then it is easy to see that the union is a finite prefix (observe
that $\union \cdot$ ``discards'' the original conflict relation and substitutes
it for a new one).

The inductive step reduces to showing that the union of two prefixes is a
prefix, as
\[
\union{F' \cup \set \ppref} = \union{\union{F'} \cup \set \ppref}.
\]
To show this, let
$\ppref_1 \eqdef \tup{E_1, <_1, \cfl_1, h_1}$
and
$\ppref_2 \eqdef \tup{E_2, <_2, \cfl_2, h_2}$
be two unfolding prefixes.
To show that $\union{\set{\ppref_1, \ppref_2}}$ is an unfolding prefix
we proceed again by induction in the size
$m$ of $E_2 \setminus E_1$.
If $m=0$ then $\ppref_2 \ispref \ppref_1$ and we are done.
If not one can select a $<$-maximal event $e \eqdef \tup{t,H}$ from $E_2 \setminus E_1$,
remove it from $\ppref_2$, and the resulting prefix
$\ppref'_2$ is such that
$\ppref_3 \eqdef \union{\ppref_1, \ppref'_2}$ is a finite prefix generated by
\cref{d:finite.prefix}.
Now \cref{d:finite.prefix} can extend $\ppref_3$ with~$e$,
as $H$ is by hypothesis a configuration of $\ppref_3$ that enables $t$ and so
on.
Finally, one shows that the causality, label, and conflict relation that
\cref{d:finite.prefix} and the definition of $\union \cdot$ will attach to~$e$
coincide.
\end{proof}

Next we show that every configuration of every unfolding prefix
corresponds to some Mazurkiewicz trace of the system:

\begin{lem}
\label{l:soundun}
Let $\ppref$ be an unfolding prefix of~$M$ under~$\indep$.
Given any configuration~$C$ of~$\ppref$,
it holds that $\inter C \subseteq \runs M$.
Furthermore,
for any two runs $\sigma_1, \sigma_2 \in \inter C$,
we have $\state{\sigma_1} = \state{\sigma_2}$.
\end{lem}

\begin{proof}
Let $\ppref \eqdef \tup{E, <, {\cfl}, h}$ be the prefix,
with $h \colon E \to T$.
Let~$C$ be a configuration of $\ppref$.
In this proof we will assume that $\ppref$ is finite.
This is because, of the following two facts:
\begin{itemize}
\item
  Assume that $\ppref = \union F$,
  where $F \eqdef \set{\ppref_1, \ppref_2, \ldots}$
  is an infinite collection of finite prefixes.
  Only finitely many prefixes in~$F$ contain events of~$C$, as $C$ is finite.
\item
  By \cref{l:f.finite}, the $\union \cdot$ of finitely many prefixes is a finite
  prefix generated by \cref{d:finite.prefix}.
\end{itemize}
So if $\ppref$ is infinite, by the above, we can find a finite prefix $\ppref'$,
generated by \cref{d:finite.prefix}, and which contains~$C$.
Since the arguments we make in the sequel only concern events in~$C$,
proving the lemma in $\ppref'$ is equivalent to proving it in $\ppref$.

So \wlogg we assume that $\ppref$ is a finite unfolding.
The proof is by structural induction on the set of unfolding prefixes
ordered by the prefix relation $\ispref$.

\emph{Base case}.
Assume that $\ppref$ has been produced by the first rule of
\cref{d:finite.prefix}.
Then $E = \set{\bot}$ and the lemma trivially holds.

\emph{Inductive step}.
Assume $\ppref$ that has been produced by the application of the second rule of
\cref{d:finite.prefix} to the unfolding prefix $\ppref'$,
and let $e$ be the only event in $\ppref$ but not in $\ppref'$.
Also, assume that the lemma holds for $\ppref'$.

Only two things are possible:
$e \in C$ or $e \notin C$.
In the second case, $C$ is a configuration of $\ppref'$ and we are done,
so assume that $e \in C$.
Necessarily $e$ is a \mbox{$<$-maximal} event in $C$.
Let $\sigma \in \inter C$ be an interleaving of $C$,
and let $C \eqdef \set{e_1, \ldots, e_n}$.
\Wlog, assume that $\sigma$ is of the form
\[
\sigma = h(e_1), \ldots, h(e_n)
\]
and that $e_i = e$.
Clearly, the causes $\causes e$ of $e$ are a subset of the events
$\set{e_1, \ldots, e_{i-1}}$.
Since, by definition of $\inter \cdot$,
$\set{e_1, \ldots, e_{i-1}}$ is a configuration and it does not
include~$e$, it is necessarily a configuration of $\ppref'$.
Thus, by applying the induction hypothesis we know that the sequence
\[
h(e_1), \ldots, h (e_{i-1})
\]
is an execution of $M$ and produces the same global state as another
execution that first fires all events in $\causes e$ and then all remaining
events in 
$\set{e_1, \ldots, e_{i-1}}$.
This means that $\sigma$ is an execution of $M$ iff
the sequence
\[
\sigma' \eqdef \sigma'' . h(f_1) \ldots h(f_k) . h(e) . h(g_1) \ldots h(g_l)
\]
is an execution of $M$,
where $\sigma'' \in \inter{\causes e}$,
$\set{f_1, \ldots, f_k} = \set{e_1, \ldots, e_{i-1}} \setminus \causes e$,
and
$g_1 = e_{i+1}$, \ldots, $g_l = e_n$.

Now we will show that the sequence
$\sigma'' . h(f_1) \ldots h(f_k) . h(e)$,
which is a prefix of $\sigma'$,
is an execution.
From \cref{d:finite.prefix} we know that $\sigma''$ enables $h(e)$, and from the
induction hypothesis we also know that $\sigma''$ enables $h(f_1)$.
Since $\lnot f_1 \cfl e$ and $f_1 \notin \causes e$,
from \cref{d:finite.prefix} we know that
$h(f_1) \indep h(e)$, \ie, the transitions associated to both events
commute (at all states).
Since both $h(f_1)$ and $h(e)$ are enabled at $\state{\sigma''}$, then
$\sigma'' . h(f_1) . h(e)$ is a run.
Again, the run
$\sigma'' . h(f_1)$ enables both $h(e)$ and $h(f_2)$, and for similar
reasons $h(e) \indep h(f_2)$, so we know that
$\sigma'' . h(f_1) . h(f_2) . h(e)$ is a run.
Iterating this argument $k$ times one can prove that
\[
\tilde \sigma \eqdef \sigma'' . h(f_1) \ldots h(f_k) . h(e)
\]
is indeed an execution.

The next step is proving that the execution $\tilde \sigma$ can be
continued by firing the sequence of transitions
$h(g_1), \ldots, h(g_l)$.
The argument here is quite similar as before, but slightly different.
It is easy to see that $h(e) \indep h(g_j)$ for
$j \in \set{1, \ldots, l}$.
Since $\tilde \sigma$ enables both $h(e)$ and $h(g_1)$, and both commute at
$\state{\tilde \sigma}$, then necessarily
$\tilde \sigma . h(e) . h(g_1)$ is an execution and reaches the same state as
the execution
$\tilde \sigma . h(g_1) . h(e)$.
Iterating this argument $l$ times one can show that, similarly,
$\tilde \sigma . h(e) . h(g_1) \ldots h(g_l)$ is an execution and reaches
the same state as the execution
$\tilde \sigma . h(g_1) \ldots h(g_l) . h(e)$.
This has shown that $\sigma$ is indeed an execution.

The lemma also requires to prove that any two executions in $\inter C$
reach the same state.
This is straightforward to show using the arguments we have introduced
above.
We have already shown that any linearization of all events in $C$
is \mbox{$h$-labelled} by an execution of $M$ that reaches the same state
as the
execution that labels any other linearization of the same events that fires
$e$ last in the sequence.
Using this fact and the induction hypothesis it is very simple to complete
the proof.
\end{proof}

\begin{lem}
   \label{l:complete.lattice}
   For any set $F$ of unfolding prefixes, $\union F$ is the least-upper bound
   of~$F$ with respect to the order $\ispref$.
\end{lem}

\begin{proof}
   Let $F \eqdef \bigcup_{1 \le i} \ppref_i$,
   where
   $\ppref_i \eqdef \tup{E_i, <_i, \cfl_i, h_i}$ for $1 \le i$.
   Let $\ppref \eqdef \union F$ be their union, where
   $\ppref \eqdef \tup{E, <, \cfl, h}$.
   We need to show that
   \begin{itemize}
   \item (upper bound) $\ppref_i \ispref \ppref$;
   \item (least element) for any unfolding prefix $\ppref'$ such that $\ppref_j \ispref \ppref'$
   holds for all $1 \le j$,
   we have that $\ppref \ispref \ppref'$.
   \end{itemize}

   We start showing that $\ppref$ is an upper bound. Let $\ppref_i \in F$ be an
   arbitrary unfolding prefix. We show that $\ppref_i \ispref \ppref$:
   \begin{itemize}
   \item
     Trivially $E_i \subseteq E$.

   \item
     ${<_i} \subseteq {<} \cap (E_i \times E_i)$. Trivial.
   \item
     ${<_i} \supseteq {<} \cap (E_i \times E_i)$.
     Assume that $e < e'$ and that both $e$ and $e'$ are in $E_i$.
     Then there is some $1 \le j$ such that $e <_j e'$,
     and both $e$ and $e'$ are in $E_j$.
     Assume that $e \eqdef \tup{t,H}$.
     Since $\ppref_j$ is a finite prefix
     constructed by \cref{d:finite.prefix},
     then necessarily $e' \in H$.
     As a result, \cref{d:finite.prefix} must have found that
     $e'$ was in $H$ when adding $e$ to the prefix that eventually became
     $\ppref_i$, and consequently $e' < e$.

   \item
     ${\cfl_i} \subseteq {\cfl} \cap (E_i \times E_i)$. Trivial.
   \item
     ${\cfl_i} \supseteq {\cfl} \cap (E_i \times E_i)$.
     Assume that $e \cfl e'$ and that $e, e' \in E_i$.
     We need to prove that $e \cfl_i e'$.
     Assume \wlogg that $e'$ was added to $\ppref_i$ by \cref{d:finite.prefix}
     after~$e$.
     If $e$ and $e'$ satisfy \cref{e:union}, then trivially
     $e \cfl_i e'$.
     If not, then assume \wlogg that
     there exists some $e'' < e'$ such that
     $e \cfl e''$, and such that $e$ and $e''$ satisfy \cref{e:union}.
     Then $e \cfl_i e''$ and,
     since $\ppref_i$ is a LES 
     then we have $e \cfl_i e'$.
   \item
     $h_i = h \cap (E_i \times E_i)$. Trivial.
   \end{itemize}

   We now focus on proving that $\ppref$ is the least element among the upper
   bounds of $F$. Let $\ppref' \eqdef \tup{E', <', \cfl', h'}$ be an upper bound of all elements of $F$.
   We show that $\ppref \ispref \ppref'$.
   \begin{itemize}
   \item
     Since $E$ is the union of all $E_i$ and all $E_i$ are by hypothesis
     in~$E'$, then necessarily $E \subseteq E'$.

   \item
     ${<} \subseteq {<'} \cap (E \times E)$.
     Assume that $e < e'$. By definition $e$ and $e'$ are in $E$, so we only
     need to show that $e <' e'$.
     We know that there is some $1 \le i$ such that $e <_i e'$.
     We also know that $\ppref_i \ispref \ppref'$, which implies that $e <' e'$.
   \item
     ${<} \supseteq {<'} \cap (E \times E)$.
     Assume that $e < e'$ and that $e,e' \in E$.
     We know that there is some $1 \le i$ such that $e,e' \in E_i$.
     We also know that $\ppref_i \ispref \ppref'$, which implies that
     ${<_i} = {<'} \cap (E_i \times E_i)$.
     This means that $e <_i e'$, and so $e < e'$.

   \item
     $h = h' \cap (E \times E)$. Trivial.

   \item
     ${\cfl} \subseteq {\cfl'} \cap (E \times E)$.
     Assume that $e \cfl e'$.
     Then $e$ and $e'$ are in $E$.
     Two things are possible. Either $e,e'$ satisfy \cref{e:union} or, \wlogg,
     there exists some $e'' < e'$ such that $e$ and $e''$ satisfy
     \cref{e:union}.
     In the former case, using items above, it is trivial to show
     that $\lnot (e <' e')$,
     that $\lnot (e' <' e)$, and
     that $h'(e) \depen h'(e')$.
     This means that $e \cfl' e'$.
     In the latter case its the same.

   \item
     ${\cfl} \supseteq {\cfl'} \cap (E \times E)$. Trivial.
   \end{itemize}
\end{proof}

\completeun*
\begin{proof}
Let $F$ be the set of all, finite or infinite,
unfolding prefixes of $\unf{M,\indep}$.
By \cref{d:prefix} we have that $\unf{M,\indep} \eqdef \union F$
is an unfolding prefix.
By \cref{l:complete.lattice} we know it is $\ispref$-maximal and unique.

Observe that for a run that fires no transition, \ie $\sigma = \varepsilon
\in T^*$, we may find the empty configuration $\emptyset$ or the
configuration $\set{\bot}$, and in both cases $\sigma$ is an
interleaving of the configuration.
Hence the restriction to non-empty runs.

Assume that $\sigma$ fires at least one transition.
The proof is by induction on the length $|\sigma|$ of the run.

\emph{Base Case}.
If $\sigma$ fires one transition~$t$, then~$t$ is enabled at $\tilde s$,
the initial state of~$M$.
Then $\set \bot$ is a history for~$t$,
as necessarily $\state{\set \bot}$ enables~$t$.
This means that $e \eqdef \tup{t, \set \bot}$ is an event of~$\unf{M,\indep}$, and
clearly $\sigma \in \inter{\set{\bot, e}}$.
It is easy to see that no other event $e'$ different than $e$ but such that
$h(e) = h(e')$ can exist in $\unf{M,\indep}$ and satisfy that the history $\causes{e'}$
of~$e'$ equals the singleton $\set \bot$.
The representative configuration for $\sigma$ is therefore unique.

\emph{Inductive Step}.
Consider $\sigma = \sigma'.t_{k+1}$, with $\sigma' = t_1.t_2 \ldots t_k$.
By the induction hypothesis, we assume that there exist a unique configuration $C'$ such that $\sigma' \in \inter{C'}   $.
By \cref{l:soundun}, all runs in $\inter{C'}$ reach the same state $s$ and $\sigma'$ is such a run.
Hence, $t_{k+1}$ is enabled at state $s$. If all $<$-maximal events $e \in
\max(C'): h(e)$ interfere with $t_{k+1}$, then
$C'$ is a valid configuration $H$ and by construction (second condition of
\cref{d:finite.prefix}) there is a configuration $C=C' \cup \set{e'}$ with
$e'=\tup{t_{k+1},H}$.
Otherwise, we construct a valid $H$ by considering sub-configurations of $C'$
removing a maximal event $e \in \max(C'): h(e)$
does not interfere with $t_{k+1}$.
We always reach a valid $H$ since $C'$ is a finite set and $\{\bot\}$ is always a valid $H$.
Considering $C=H \cup \{e'\}$ with $e'=\tup{t,H}$, by construction (second
condition of \cref{d:finite.prefix}) we have that $\forall e_H \in H: \neg(e'\cfl e_H)$ and
$\forall e_{\hat H} \in C'\setminus H: \neg(e\cfl e_{\hat H})$ (otherwise these
events
would be in $H$). Hence, $C' \cup \{e\}$ is a configuration.
\end{proof}

\section{Proofs: Exploration Algorithm}
\label{x:lazy}

\newcommand\calls {\mathrel{\triangleright}}
\newcommand\callsl {\mathrel{\triangleright_l}}
\newcommand\callsr {\mathrel{\triangleright_r}}


For the rest of this section, as we did in the main sections of the paper,
we fix a system~$M \eqdef \tup{\Sigma, T, \tilde s}$
and an unconditoinal independence relation $\indep$ on~$M$.
We assume that $\reach M$ is finite.
Let $\unf{\indep,M} \eqdef \tup{E, <, {\cfl}, h}$ be the unfolding of~$M$
under~$\indep$, which we abbreviate as~$\uunf$.
For this section, unless otherwise state, we furthermore assume that
that $\uunf$ is finite, \ie, that all computations of~$M$ terminate.

\Cref{a:a1} is recursive, each call to \explore{$C,D,A$} 
yields either no recursive call, if the function returns at
\cref{l:a1ret},
or one single recursive call (\cref{l:a1l}),
or two (\cref{l:a1l} and \cref{l:a1r}).
Furthermore, it is non-deterministic, as $e$ is chosen from either the set
$\en C$ or the set $A \cap \en C$, which in general are not singletons.
As a result, the configurations explored by it
may differ from one execution to the next.

For each system $M$ we define the \emph{call graph} explored by
\cref{a:a1} as a directed graph $\tup{B, {\calls}}$ representing the actual
exploration that the algorithm did on the state space. Different executions
will in general yield different call graphs.

The nodes~$B$ of the call graph are
4-tuples of the form $\tup{C,D,A,e}$, where $C,D,A$ are the parameters of
a recursive call made to the funtion \explore{$\cdot,\cdot,\cdot$}, and
$e$ is the event selected by the algorithm immediately before \cref{l:a1l}.
More formally,~$B$ contains exactly all tuples $\tup{C,D,A,e}$ satisfying
that
\begin{itemize}
\item
   $C$, $D$, and $A$ are sets of events of the unfolding $\uunf$;
\item
   during the execution of \explore{$\emptyset,\emptyset,\emptyset$},
   the function \explore{$\cdot,\cdot,\cdot$} has been recursively called
   with $C,D,A$ as, respectively, first, second, and third argument;
\item
   $e \in E$ is the event selected by \explore{$C,D,A$} immediately
   before \cref{l:a1l} if $C$ is not maximal; if $C$ is
   maximal, we define $e \eqdef \bot$.
   \footnote{Observe that in this case, if $C$ is maximal, the execution of
   \explore{$C,D,A$} never reaches \cref{l:a1l}.}
\end{itemize} 
The edge relation of the call graph,
${\calls} \subseteq B \times B$, represents the recursive calls made by
\explore{$\cdot,\cdot,\cdot$}.
Formally, it is the union of two disjoint relations
${\calls} \eqdef {\callsl} \uplus {\callsr}$, defined as follows.
We define that
\[
\tup{C,D,A,e} \callsl \tup{C',D',A',e'}
\text{ ~~~ and that ~~~ }
\tup{C,D,A,e} \callsr \tup{C'',D'',A'',e''}
\]
iff the execution of \explore{$C,D,A$} issues a recursive
call to, \resp,
\explore{$C',D',A'$} at \cref{l:a1l} and
\explore{$C'',D'',A''$} at \cref{l:a1r}.
Observe that $C'$ and $C''$ will necessarily be different
(as $C' = C \cup \set e$, where $e \notin C$, and $C'' = C$), and therefore the
two relations are disjoint sets.
We distinguish the node
\[
b_0 \eqdef \tup{\set\bot,\emptyset,\emptyset,\bot}
\]
as the \emph{initial node}, also called the \emph{root node}.
Observe that $\tup{B, {\calls}}$ is by definition a weakly connected
digraph, as there is a path from the node $b_0$ to every other node in~$B$.
Later in this section we will additionally prove that the call graph is
actually a binary tree, where $\callsl$ is the \textit{left-child} relation
and $\callsr$ is the \textit{right child} relation.


\subsection{General Lemmas}
\label{x:general}

\begin{lem}
\label{l:general}
Let $\tup{C,D,A,e} \in B$ be a state of the call graph. We have that
\begin{itemize}
\item
   event $e$ is such that $e \in \en C$;
   \eqtag{e:basic1}
\item
   $C$ is a configuration;
   \eqtag{e:basic2}
\item
   $C \cup A$ is a configuration and $C \cap A = \emptyset$;
   \eqtag{e:basic3}
\item
   $D \subseteq \ex C$;
   \eqtag{e:basic4}
\item
   if $A = \emptyset$, then $D \subseteq \cex C$;
   \eqtag{e:basic5}
\item
   for all $e' \in D$ there is some $e'' \in C \cup A$ such that
   $e' \icfl{} e''$
   \eqtag{e:basic6}
\end{itemize}
\end{lem}
\begin{proof}
   To show~\cref{e:basic1} is immediate.
   Observe, in \cref{a:a1},
   that both branches of the ``\emph{if}'' statement where $e$ is picked
   select it from the set~$\en C$.

   All remaining items,
   \cref{e:basic2,e:basic3,e:basic4,e:basic5,e:basic6}, will be
   shown by induction on the length $n \ge 0$ of any path
   \[
   b_0 \calls b_1 \calls \ldots \calls b_{n-1} \calls b_n
   \]
   on the call graph, starting from the initial node
   and leading to 
   $b_n \eqdef \tup{C, D, A, e}$
   (we will later show, \cref{l:c3c4}, that there is actually only one such path).
   For $i \in \set{0, \ldots, n}$ we define
   $\tup{C_i,D_i,A_i,e_i} \eqdef b_i$.

   We start showing \cref{e:basic2}.
   \emph{Base case.}
   $n = 0$ and $C = \set\bot$. The set $\set\bot$ is a configuration.
   \emph{Step.}
   Assume $C_{n-1}$ is a configuration. If
   $b_{n-1} \callsl b_n$,
   then $C = C_{n-1} \cup \set e$ for some event $e \in \en C$, as stated
   in \cref{e:basic1}. By definition, 
   $C$ is a configuration.
   If 
   $b_{n-1} \callsr b_n$,
   then $C = C_{n-1}$. In any case $C$ is a configuration.

   We show \cref{e:basic3}, also by induction on $n$.
   \emph{Base case.}
   $n = 0$. Then $C = \set\bot$ and $A = \emptyset$. Clearly $C \cup A$ is a
   configuration and $C \cup A = \emptyset$.
   \emph{Step.}
   Assume that $C_{n-1} \cup A_{n-1}$ is a configuration and that
   $C_{n-1} \cap A_{n-1} = \emptyset$.
   We have two cases.
   \begin{itemize}
   \item
      Assume that $b_{n-1} \callsl b_n$.
      If $A_{n-1}$ is empty, then $A$ is empty as well. Clearly
      $C \cup A$ is a configuration and $C \cap A$ is empty.
      If $A_{n-1}$ is not empty, then
      $C = C_{n-1} \cup \set e$ and
      $A = A_{n-1} \setminus \set e$, for some
      $e \in A_{n-1}$, and we have
      \[
         C \cup A =
         (C_{n-1} \cup \set e ) \cup (A_{n-1} \setminus \set e ) =
         C_{n-1} \cup A_{n-1},
      \]
      so $C \cup A$ is a configuration as well.
      We also have that $C \cap A = C_{n-1} \cap A_{n-1}$ (recall that
      $e \notin C$), so $C \cap A$ is empty.
   \item
      Assume that $b_{n-1} \callsr b_n$ holds.
      Then we have $C = C_{n-1}$ and also
      $A = J \setminus C_{n-1}$ for some
      $J \in \alternatives{$C_{n-1},D \cup \set e$}$.
      From \cref{e:alt2}
      we know that $C_{n-1} \cup J$ is a configuration.
      As a result,
      \[
         C \cup A =
         C_{n-1} \cup (J \setminus C_{n-1}) =
         C_{n-1} \cup J,
      \]
      and therefore $C \cup A$ is a configuration.
      Finally, by construction of $A$, we clearly have
      $C \cap A = \emptyset$.
   \end{itemize}

   We show \cref{e:basic4}, again, by induction on $n$.
   \emph{Base case}.
   $n = 0$ and $D = \emptyset$. Then \cref{e:basic4} clearly holds.
   \emph{Step.}
   Assume that \cref{e:basic4} holds for $\tup{C_i,D_i,A_i,e_i}$, with
   $i \in \set{0, \ldots, n - 1}$. We show that it holds for $b_n$.
   As before, we have two cases.
   \begin{itemize}
   \item
      Assume that $b_{n-1} \callsl b_n$.
      We have that
      $D = D_{n-1}$ and that
      $C = C_{n-1} \cup \set{e_{n-1}}$.
      We need to show that for all $e' \in D$
      we have $\causes{e'} \subseteq C$ and $e' \notin C$.
      By induction hypothesis we know that
      $D = D_{n-1} \subseteq \ex{C_{n-1}}$, so clearly
      $\causes{e'} \subseteq C_{n-1} \subseteq C$.
      We also have that $e' \notin C_{n-1}$, so we only need to check that
      $e' \ne e_{n-1}$.
      By contradiction,
      if $e' = e_{n-1}$, by \cref{e:basic6} we would have that some event
      in~$C$ is conflict with some other event in~$C \cup A$, which is a
      contradiction to \cref{e:basic3}.
   \item
      Assume that $b_{n-1} \callsr b_n$.
      We have that $D = D_{n-1} \cup \set{e_{n-1}}$, and by hypothesis we
      know that $D_{n-1} \subseteq \ex{C_{n-1}} = \ex C$.
      As for $e_{n-1}$, by \cref{e:basic1} we know that
      $e_{n-1} \in \en{C_{n-1}} = \en C \subseteq \ex C$.
      As a result, $D \subseteq \ex C$.
   \end{itemize}

   We show \cref{e:basic5}.
   By \cref{e:basic4} we know that $D \subseteq \ex C$.
   Assume $A = \emptyset$.
   For each $e' \in D$ we need to prove the existence of some
   $e'' \in C$ with $e' \icfl{} e''$.
   This is exactly what \cref{e:basic6} states.

   We show \cref{e:basic6}, again, by induction on $n$.
   \emph{Base case}.
   $n = 0$ and $D = \emptyset$. The result holds.
   \emph{Step}.
   Assume \cref{e:basic6} holds for
   $\tup{C_{n-1},D_{n-1},A_{n-1},e_{n-1}}$.
   We show that it holds for $b_n$.
   We distinguish two cases.
   \begin{itemize}
   \item
      $b_{n-1} \callsl b_n$.
      Then $D = D_{n-1}$. As a result,
      for any $e' \in D$ there is some $e'' \in C_{n-1} \cup A_{n-1}$
      satisfying $e' \icfl{} e''$.
      But we have that $C_{n-1} \cup A_{n-1} \subseteq C \cup A$,
      so such $e'$ is also contained in $C \cup A$, which shows the result.
   \item
      $b_{n-1} \callsr b_n$.
      Observe that $D = D_{n-1} \cup \set{e_{n-1}}$.
      Let $J \in \alternatives{$C_{n-1},D \cup \set e$}$ be the alternative used to
      construct $A = J \setminus C_{n-1}$.
      By definition \cref{e:alt1} we know that for all
      $e' \in D \setminus \cex{C_{n-1}}$ we can find some $e'' \in J$ with
      $e' \icfl{} e''$. We only need to show that $J \subseteq A \cup C$.
      Observe that this will complete the proof, since for each
      $e' \in D \cap \cex{C_{n-1}}$ we already know that there is some
      $e'' \in C_{n-1} \subseteq C \cup A$ with $e' \icfl{} e''$.
      Now, that $J \subseteq C \cup A$ is obvious:
      $C \cup A = C_{n-1} \cup J \setminus C_{n-1} = C_{n-1} \cup J$.
   \end{itemize}
\end{proof}

The following lemma essentially guarantees that whenever
\cref{a:a1} reaches \cref{l:a1choose}, the set from which $e$ is chosen is
not empty.

\begin{lem}
   \label{l:twoconfs}
   If $C \subseteq C'$ are two finite configurations, then
   $\en C \cap (C' \setminus C) = \emptyset$
   iff
   $C' \setminus C = \emptyset$.
\end{lem}
\begin{proof}
If there is some $e \in \en C \cap (C' \setminus C)$,
then $e \notin C$ and $e \in C'$, so $C' \setminus C$ is not empty.
If there is some $e' \in C' \setminus C$, then there is some $e''$ event
that is $<$-minimal in $C' \setminus C$.
As a result, $\causes{e''} \subseteq C$. Since $e'' \notin C$
and $C \cup \set{e''}$ is a configuration
(as $C\cup \set{e''} \subseteq C'$), we have that $e'' \in \en C$.
Then $\en C \cap (C' \setminus C)$ is not empty.
\end{proof}

\begin{lem}
   For any node $\tup{C,D,A,e} \in N$ of the call graph we have that
   $A \ne \emptyset$ implies $\en C \cap A \ne \emptyset$.
\end{lem}
\begin{proof}
   The result is a consequence of \cref{l:twoconfs} and \cref{e:basic3}.
   Since $C \cup A$ is configuration that includes $C$,
   and $(C \cup A) \setminus C = A$ is not empty,
   then $\en C \cap A$ is not empty.
\end{proof}

\begin{lem}
   Let $b \eqdef \tup{C,D,A,e}$ and $b' \eqdef \tup{C',D',A',e'}$ be
   two nodes of the call graph such that $b \calls b'$.
   Then
   \begin{itemize}
      \item $C \subseteq C'$ and $D \subseteq D'$; \eqtag{e:step1}
      \item if $b \callsl b'$, then $C \subsetneq C'$; \eqtag{e:step2}
      \item if $b \callsr b'$, then $D \subsetneq D'$. \eqtag{e:step3}
   \end{itemize}
\end{lem}
\begin{proof}
   If $b \callsl b'$, then
   $C' = C \cup \set e$ and
   $D' = D$.
   Then all the three statements hold.
   If $b \callsr b'$, then
   $C' = C$ and
   $D' = D \cup \set e$.
   Similarly, all the three statements hold.
\end{proof}

\subsection{Termination}
\label{x:termination}

\begin{lem}
   \label{l:finitepath}
   Any path $b_0 \calls b_1 \calls b_2 \calls \ldots$
   in the call graph starting from $b_0$ is finite.
\end{lem}
\begin{proof}
   By contradiction.
   Assume that $b_0 \calls b_1 \calls \ldots$ is an infinite path in the
   call graph.
   For $0 \le i$, let $\tup{C_i,D_i,A_i,e_i} \eqdef b_i$.
   Recall that $\uunf$ has finitely many events, finitely many finite
   configurations, and no infinite configuration.
   Now, observe that the number of times that $C_i$ and $C_{i+1}$ are
   related by $\callsl$ rather than $\callsr$ is finite, since every time
   \explore{$\cdot,\cdot,\cdot$} makes a recursive call at \cref{l:a1l} it
   adds one event to $C_i$, as stated by \cref{e:step2}.
   More formally, the set
   \[ L \eqdef \set{i \in \N \colon C_i \callsl C_{i+1}} \]
   is finite.
   As a result it has a maximum, and its successor
   $k \eqdef 1 + \max_< L$ is an index in the path such that
   for all $i \ge k$ we have $C_i \callsr C_{i+1}$, \ie, the function only
   makes recursive calls at \cref{l:a1r}.
   We then have that $C_i = C_k$, for $i \ge k$, and by \cref{e:basic4},
   that $D_i \subseteq \ex{C_k}$.
   Recall that $\ex{C_k}$ is finite. Observe that, as a result of \cref{e:step2},
   the sequence
   \[
   D_k \subsetneq D_{k+1} \subsetneq D_{k+2} \subsetneq \ldots
   \]
   is an infinite increasing sequence.
   This is a contradiction, as for sufficiently large $j \ge 0$ we will
   have that $D_{k+j}$ will be larger than $\ex{C_k}$, yet
   $D_{k+j} \subseteq \ex{C_k}$.
\end{proof}

\begin{cor}
   \label{c:dag}
   The call graph is a finite directed acyclic graph.
\end{cor}
\begin{proof}
   Recall that every node $b \in B$ is reachable from the initial node
   $b_0$ by definition of the graph.
   Also, by \cref{l:finitepath}, all paths from $b_0$ are finite, and
   every node has between~0 and~2 adjacent nodes.

   By contradiction, if the graph had infinitely may nodes, then König's
   lemma would guarantee the existence of an infinite path starting from
   $b_0$, a contradiction to \cref{l:finitepath}.
   Then $B$ is necessarily finite.

   As for the acyclicity, again by contradiction, assume that
   $\tup{B, {\calls}}$ has a cycle. Then every state of any such cycle
   would be reachable from $b_0$, which guarantees the existence of at
   least one infinite path in the graph. Again, this is a contradiction to
   \cref{l:finitepath}.
\end{proof}

\thmtermination*
\begin{proof}
   Remark that \cref{a:a1} makes calls to three functions, namely,
   \extend{$\cdot$},
   \remove{$\cdot$},
   and
   \alternatives{$\cdot, \cdot$},
   Clearly the first two terminate. Since we gave no algorithm to compute
   \alternatives{$\cdot$}, we will assume we employ one that
   terminates on every input.

   Now, observe that there is no loop in \cref{a:a1}.
   Thus any non-terminating execution of \cref{a:a1} must perform a
   non-terminating sequence of recursive calls, which entails the existence
   of an infinite path in the call graph associated to the execution.
   Since, by \cref{l:finitepath}, no infinite path exist in the call graph,
   \cref{a:a1} always terminates.
\end{proof}

\subsection{Optimality}
\label{x:optimality}

\begin{lem}
   \label{l:c3c4}
   Let $b, b_1, b_2, b_3, b_4 \in B$ be nodes of the call graph such that
   \[
      b \callsl b_1 \calls^* b_3
      \text { ~~ and ~~ }
      b \callsr b_2 \calls^* b_4.
   \]
   and such that $\tup{C_3, D_3, A_3, e_3} \eqdef b_3$ and
   $\tup{C_4, D_4, A_4, e_4} \eqdef b_4$.
   Then $C_3 \ne C_4$.
\end{lem}
\begin{proof}
   Let $\tup{C, D, A, e} \eqdef b$,
   $\tup{C_1, D_1, A_1, e_1} \eqdef b_1$, and
   $\tup{C_2, D_2, A_2, e_2} \eqdef b_2$.
   By \cref{e:step2} we know that $e \in C_1$, and by \cref{e:step1} that
   $e \in C_3$. We show that $e \notin C_4$.
   By \cref{e:step3} we have that $e \in D_2$, and again by \cref{e:step1}
   that $e \in D_4$. Since $D_4 \subseteq \ex{C_4}$, by \cref{e:basic4}, we
   have that $e \in \ex{C_4}$, so $e \notin C_4$.
\end{proof}

\begin{cor}
   \label{c:btree}
   The call graph $(B, {\calls})$ is a finite binary tree,
   where $\callsl$ and $\callsr$
   are respectively the \emph{left-child} and \emph{right-child} relations.
\end{cor}
\begin{proof}
   \Cref{c:dag} states that the call graph is a finite directed acyclic
   graph. \Cref{l:c3c4} guarantees that for every node $b \in B$, the nodes
   reached after the left child are different from those reached after the
   right one.
\end{proof}

\begin{lem}
   \label{l:atmost1}
   For any maximal configuration $C \subseteq E$,
   there is \emph{at most one} node
   $\tup{\tilde C, \tilde D, \tilde A, \tilde e} \in B$
   with $C = \tilde C$.
\end{lem}
\begin{proof}
   By contradiction, assume there was two different nodes,
   \[
   \hat b \eqdef \tup{C, \hat D, \hat A, \hat e}
   \text{ ~~ and ~~ }
   b' \eqdef \tup{C,D',A',e'}
   \]
   in $B$ such that the first component of the tuple is $C$.
   The call graph is a binary tree, because of \cref{c:btree}, so there is
   exactly one path from
   $b_0 \eqdef \tup{\emptyset, \emptyset, \emptyset, e_0}$
   to respectively $\hat b$ and $b'$.
   Let
   \[
   \hat b_0 \calls \hat b_1 \calls \ldots \calls \hat b_{n-1} \calls \hat b_n
   \text{ ~~ and ~~ }
   b'_0 \calls b'_1 \calls \ldots \calls b'_{m-1} \calls b'_m
   \]
   be the two such unique paths,
   with $\hat b_n \eqdef \hat b$, $b'_n \eqdef b'$ and
   $\hat b_0 \eqdef b'_0 \eqdef b_0$.
   Such paths clearly share the first node $b_0$. In general they will
   share a number of nodes to later diverge. Let $i$ be the index of the
   last node common to both paths, \ie, the maximum integer $i \ge 0$ such
   that
   \[
   \tup{\hat b_0, \hat b_1, \ldots, \hat b_i}
   =
   \tup{b'_0, b'_1, \ldots, b'_i}
   \]
   holds.
   Observe both paths necessarily diverge before reaching the last node,
   \ie, one cannot be a prefix of the other.
   This is because both $\hat b$ and $b'$ are leaves of the call graph,
   \ie, there is no $b'' \in B$ such that either $\hat b \calls b''$ or
   $b' \calls b''$.
   As a result
   $\hat b \ne b'_j$ for any $j \in \set{0, \ldots, m}$
   and
   $b' \ne \hat b_j$ for any $j \in \set{0, \ldots, n}$.
   This means that $i < \min \set{n, m}$.

   Let $\tup{C_i,D_i,A_i,e_i} \eqdef b_i$.
   \Wlog, assume that
   $\hat b_i \callsl \hat b_{i+1}$ and that
   $b'_i \callsr b'_{i+1}$.
   Now, using \cref{e:step2} and \cref{e:step1}, it is simple to show that
   $e_i \in C$.
   And using \cref{e:step3} and \cref{e:step1}, that
   $e_i \in D'$.
   Then, by \cref{e:basic4} we get that $e_i \in \ex C$,
   a contradiction to $e_i \in C$.
\end{proof}

\thmoptimality*
\begin{proof}
   By construction, every call to \explore{$C,D,A$} produces one node of
   the form $\tup{C,D,A,e}$, for some~$e \in E$, in the call graph
   associated to the execution.
   By \cref{l:atmost1}, there is at most one node with its first parameter
   being equal to~$\tilde C$, so \explore{$\cdot,\cdot,\cdot$} can have
   been called at most once with~$\tilde C$ as first parameter.

   Observe, furthermore, that the algorithm does not initiate what
   Abdulla et~\al call \emph{sleep-set blocked executions}~\cite{AAJS14}.
   These correspond, in our setting, to exploring the same configuration in both
   branches of the tree. Formally, our algorithm would explore sleep-set blocked
   executions iff it is possible to find some $b \in B$ such that the left and
   right subtrees of~$b$ contain nodes exploring the same configuration.
   By \cref{l:c3c4} this is not possible.
\end{proof}

\subsection{Completeness}
\label{x:completeness}

\begin{lem}
   \label{l:compl.step}
   Let $b \eqdef \tup{C,D,A,e} \in B$ be a node in the call graph and
   $\hat C \subseteq E$ an arbitrary maximal configuration of $\uunf$
   such that $C \subseteq \hat C$ and $D \cap \hat C = \emptyset$.
   Then exactly one of the following statements hold:
   \begin{itemize}
   \item
     Either $C$ is a maximal configuration of $\uunf$, or
   \item
     $e \in \hat C$ and~$b$ has a left child, or
   \item
     $e \notin \hat C$ and~$b$ has a right child.
   \end{itemize}
\end{lem}
\begin{proof}
   The proof is by induction on~$b$ using a specific total order in~$B$
   that we define now.
   Recall that $\tup{B,{\calls}}$ is a binary tree (\cref{c:btree}).
   We let ${\lessdot} \subseteq B \times B$ be the unique
   in-order relation in~$B$. Formally, $\lessdot$ is
   the order that sorts, for every~$\tilde b \in B$, first all nodes
   reachable from~$\tilde b$'s left child (if there is any),
   then~$\tilde b$, then all nodes reachable from~$\tilde b$'s right child
   (if there is any).

   \emph{Base case.}
   Node~$b$ is the least element in~$B$ \wrt $\lessdot$. Then~$b$ is
   the leftmost leaf of the call tree, \ie,
   $b_0 \callsl^* b$,
   and~$C$ is a maximal configuration.
   Then the first item holds.

   \emph{Step case.}
   Assume that the result holds for any node $\tilde b \lessdot b$.
   If~$C$ is maximal, we are done. 
   So assume that $C$ is not maximal, and so that~$b$ has at least one left
   child.
   If $e \in \hat C$, then we are done, as the second item holds.

   So assume that $e \notin \hat C$.
   The rest of this proof shows that the third item of the lemma holds,
   \ie, that~$b$ has right child. In particular we show that
   there exists some alternative $\hat J \subseteq \hat C$ such that
   $\hat J \in \alternatives{$C, D \cup \set e$}$.

   We start by setting up some notation.
   Observe that any alternative
   $J \in \alternatives{$C, D \cup \set e$}$ needs
   to contain, for every event $e' \in D \cup \set e$,
   some event $e'' \in J \cup C$ in immediate conflict with $e'$,
   \cf \cref{e:alt1}.
   In fact $e''$ can be in~$J$ or in~$C$.
   Those $e' \in D \cup \set e$ such that $C$ already contains some~$e''$
   in conflict with $e'$ pose no problem. So we need to focus on the
   remaining ones, we assign them a specific name,
   we define the set
   \[
   F \eqdef \set{e_1, \ldots, e_n} \eqdef D \setminus \cex C \cup \set e.
   \]
   Let $e_i$ be any event in~$F$.
   Clearly $e_i \in \cex{\hat C}$, as
   $e_i \in D \subseteq \ex C$, by \cref{e:basic4}, and so
   $\causes{e_i} \subseteq C \subseteq \hat C$ and
   $e_i \notin \hat C$.
   Since $e_i \in \cex{\hat C}$ we can find some $e'_i \in \hat C$
   such that $e_i \icfl e'_i$.
   We can now define a set
   \[
   \hat J \eqdef [\set{e'_1, \ldots, e'_n}]
   \]
   such that
   $e'_i \in \hat C$ and $e_i \icfl e'_i$
   for $i \in \set{1, \ldots, n}$.
   Clearly $\hat J \subseteq \hat C$ and $\hat J$ is causally closed,
   so it is a configuration.
   Observe that $\hat J$ is not uniquely defined, there may be
   several~$e'_i$ to choose for each~$e_i$ (some of the~$e'_i$ might even
   be the same).
   We take any~$e'_i$ in immediate conflict with~$e_i$, the choice is
   irrelevant (for now).

   We show now that $\hat J \in \alternatives{$C, D \cup \set e$}$
   when function \alternatives{$\cdot$} is called just before
   \cref{l:a1r} during the execution of \explore{$C,D,A$}.
   Let $\hat U$ be the set of events contained in variable~$U$ of
   \cref{a:a1} exactly when \alternatives{$\cdot$} is called.
   Clearly $C \cup \hat J$ is a configuration, so
   \cref{e:alt1} holds.
   To verify \cref{e:alt2},
   consider any event $\tilde e \in D \cup \set e$.
   If $\tilde e \in D \cap \cex C$ we can always find
   some~$\tilde e' \in C$
   with $\tilde e' \in \ficfl[\tilde U]{\tilde e}$.
   If not, then $\tilde e = e_i$ for some $i \in \set{1, \ldots, n}$
   and we can find some $e'_i \in \hat J$ such that
   $e_i \icfl e'_i$.
   In order to verify \cref{e:alt2} we only need to check
   that~$e'_i \in \hat U$. In the rest of this proof we show this.
   Observe that $e'_i \in \hat U$ also implies that
   $\hat J \subseteq \hat U$, necessary to ensure
   that $\hat J$ is an alternative to $D \cup \set e$ after $C$
   \emph{when the function \alternatives{$\cdot$} is called}.

   In the sequel we show that $\hat J \subseteq \hat U$.
   In other words, that event $e'_i$, for $i \in \set{1, \ldots, n}$,
   is present in set~$U$ when function \alternatives{$C, D \cup \set e$} is
   called.
   The set~$U$ has been filled with events in function \extend{$\cdot$} as
   the exploration of~$\uunf$ advanced, some of them have been kept
   in~$U$, some of them have been removed with \remove{$\cdot$}. To reason
   about the events in~$\hat U$ we need to look at fragment of $\uunf$
   explored so far.

   For $i \in \set{1, \ldots, n}$ let
   $b_i \eqdef \tup{C_i, D_i, A_i, e_i} \in B$ be the node in the call
   graph associated to event $e_i \in F$.
   These nodes are all situated in the unique path from~$b_0$ to~$b$.
   \Wlog assume (after possible reordering of the index~$i$) that
   \[
   b_0 \calls^*
   b_1 \calls^*
   b_2 \calls^*
   \ldots \calls^*
   b_n
   \]
   where $b_n = b$ and $e_n = e$.
   First observe that for any $i \in \set{2, \ldots, n}$
   we have $\set{e_1, \ldots, e_{i-1}} \subseteq D_i$.
   Since every event $e_i$ is in $D = D_n$,
   for $i \in \set{1, \ldots, n-1}$,
   we know that the first step in the path that goes
   from $b_i$ to $b_{i+1}$ is a \emph{right child}.
   In other words, the call to \explore{$C_i,D_i,A_i$} is right now blocked
   on the \emph{right-hand side recursive call} at \cref{l:a1r} in
   \cref{a:a1}, after having decided that there was one right child to
   explore.
   For the shake of clarity, we can then informally write
   \[
   b_0 \calls^*
   b_1 \mathrel{{\callsr}{\calls^*}}
   b_2 \mathrel{{\callsr}{\calls^*}}
   \ldots \mathrel{{\callsr}{\calls^*}}
   b_n.
   \]
   We additionally define the sets of events
   \[
   U_0, U_1, \ldots, U_n \subseteq E
   \]
   as, respectively for $i \in \set{1, \ldots, n}$,
   the value of the variable~$U$ during the execution of
   \explore{$C_i,D_i,A_i$} just before the
   \emph{right recursive call} at \cref{l:a1r} was made,
   \ie, the value of variable~$U$ when
   \alternatives{$C_i,D_i \cup \set{e_i}$} was called.
   For $i = 0$ we set $U_0 \eqdef \set \bot$ to
   the initial value of~$U$.
   According to this definition we have that $U_n = \hat U$.

   To prove that $\hat J \subseteq \hat U = U_n$ it is now sufficient to
   prove that $e'_i \in U_i$, for $i \in \set{1, \ldots, n}$.
   This is essentially because of the following three facts.
   \begin{enumerate}
   \item
     Clearly $e_i \in U_i$.
   \item
     For any node $\tilde b \eqdef \tup{\tilde C, \tilde D, \cdot, \tilde e} \in B$
     explored after $b_i$ and before $b_n$ it holds that~$e_i \in \tilde D$,
     by \cref{e:step1}, and so every time function
     \remove{$\tilde e, \tilde C, \tilde D$}
     has been called, event $e_i$ has not been removed from~$U$.
   \item
     Any event in immediate conflict with~$e_i$ will likewise
     not be removed from set~$U$ as long as~$e_i$ remains in~$D$,
     for the same reason as before.
   \end{enumerate}
   In other words, $e'_i \in U_i$ implies that $e'_i \in U_n$,
   for $i \in \set{1, \ldots, n}$.

   We need to show that $e'_i \in U_i$, for $i \in \set{1,\ldots,n}$.
   Consider the configuration $C' \subseteq E$ defined as follows:
   \[
   C' \eqdef C \cup \set{e_i} \cup \causes{e'_i}.
   \]
   First, note that~$C'$ is indeed a configuration, since
   it is clearly causally closed and there is no conflict:
   $e_i \in \en C$ and
   $C \cup \causes{e'_i} \subseteq \hat C$ and
   $[e_i] \cup \causes{e'_i}$ is conflict-free
   (because $e_i$ and $e'_i$ are in \emph{immediate} conflict).
   Remark also that $D_i \subseteq \ex{C'}$
   and that $e'_i \in \cex{C'}$.
   We now consider two cases:
   \begin{itemize}
   \item
     \emph{Case 1}:
     there is some maximal configuration~$C'' \supseteq C'$ such that
     $D_i \cap C'' = \emptyset$.
     We show that~$C''$ have been visited during the exploration of
     the left subtree of~$b_i$.
     In that case, since~$e'_i \in \cex{C''}$ and~$e_i \in C''$,
     \cref{a:a1} will have been appended~$e'_i$ to~$U$ during that
     exploration, and~$e'_i$ will remain in~$U$ at least as long as~$e_i$
     is in~$D$.

     To show that $C''$ has been explored,
     consider the left
     child~$b'_i \eqdef \tup{C_i \cup \set{e_i}, D_i, \cdot, \cdot}$
     of~$b_i$.
     In that case, since~$b_i \lessdot b$ (recall that $b$ is in the right
     subtree of~$b_i$), clearly every node~$\hat b \in B$ in the subtree
     rooted at $b'_i$ (\ie, $b'_i \calls^* \hat b$)
     is such that $\hat b \lessdot b_i \lessdot b$.
     This means that the induction hypothesis applies to~$\hat b$.
     So \cref{l:compl.findit} applied to~$b'_i$ and~$C''$ shows that
     $C''$ has been explored in the subtree rooted at $b'_i$.
     As a result $e'_i \in U_i$ and
     $e'_i \in U_n$, what we wanted to prove.
     
   \item
     \emph{Case 2}:
     there is no maximal configuration $C'' \supseteq C'$ such that
     $D_i \cap C'' = \emptyset$.
     In other words, \emph{any} maximal configuration $C'' \supseteq C'$
     is such that $D_i \cap C'' \ne \emptyset$.
     Our first step is showing that this implies that
     \begin{equation}
     \label{e:existsj}
     \exists j \in \set{1, \ldots, i-1}
     \text{ such that }
     \fcfl{e_i} \cap \hat C \supseteq \fcfl{e_j} \cap \hat C.
     \end{equation}
     Let $C'' \supseteq C'$ be a maximal configuration.
     Then $D_i \cap C'' \ne \emptyset$.
     This implies that $D_i \cap \en C \cap C'' \ne \emptyset$,
     as necessarily $D_i \cap C'' \subseteq \en C$.
     Observe that $D_i \cap \en C = \set{e_1, \ldots, e_{i-1}}$, so we have
     that $\set{e_1, \ldots, e_{i-1}} \cap C'' \ne \emptyset$.
     Consider now the following two sets:
     \[
     X_1 \eqdef \hat C \setminus \fcfl{e_i}
     \text{ ~~ and ~~ }
     X_2 \eqdef X_1 \cup \set{e_i}.
     \]
     Observe now the following. We can find a maximal configuration
     $C''' \supseteq X_1$ satisfying that $D_i \cap C''' = \emptyset$
     (for instance, take $C''' \eqdef \hat C$).
     But, because $C' \subseteq X_2$, we cannot find any
     $C''' \supseteq X_2$ satisfying that $D_i \cap C''' = \emptyset$.
     This implies that for any
     $C''' \supseteq X_2$ we have
     $\set{e_1, \ldots, e_{i-1}} \cap C''' \ne \emptyset$.
     Based on the last statement we can now prove \cref{e:existsj} by
     contradiction.
     Assume that \cref{e:existsj} does not hold.
     Then for any $j \in \set{1, \ldots, i-1}$, one could find some event
     $\tilde e \in \fcfl{e_j} \cap \hat C$ such that
     $\tilde e \notin \fcfl{e_i} \cap \hat C$.
     Then $\tilde e \notin \fcfl{e_i}$ and as a result
     $\tilde e \in X_1 \subseteq X_2$.
     This now would mean that for any $j \in \set{1, \ldots, i-1}$
     it holds that $\fcfl{e_j} \cap X_2 \ne \emptyset$.
     This implies that \emph{any} maximal configuration~$C'''$
     extending~$X_2$ is such that 
     $\set{e_1, \ldots, e_{i-1}} \cap C''' = \emptyset$.
     This is a contradiction,
     so the validity of \cref{e:existsj} is now established.

     According to \cref{e:existsj} there might be several
     integers $j \in \set{1, \ldots, i-1}$ such that
     $\fcfl{e_i} \cap \hat C \supseteq \fcfl{e_j} \cap \hat C$
     holds. Let $m$ be the minimum such~$j$, and consider the following
     set:
     \[
     X_3 \eqdef X_1 \cup \set{e_m} \cup \causes{e'_m}.
     \]
     We will now prove that $X_3$ is a configuration and it has been
     visited during the exploration of the subtree rooted at the left child
     of~$b_m$. We first establish several claims about~$X_3$:
     \begin{itemize}
     \item
       \emph{Fact 1:
       set $X_3$ is causally closed.}
       Since $X_1$ is causally closed, clearly
       $X_1 \cup \causes{e'_m}$ is causally closed.
       Now, since $\set{e_i, e_m} \subseteq \en C$, we have that
       $\fcfl{e_i} \cap C = \emptyset$, and as a result
       $\causes{e_m} \subseteq C \subseteq X_1 \subseteq X_3$.
     \item
       \emph{Fact 2:
       set $X_3$ is conflict free.}
       Since $X_1 \cup \causes{e'_m} \subseteq \hat C$, there is no pair of
       confliting events in $X_1 \cup \causes{e'_m}$.
       Consider now~$e_m$.
       Since~$e_m$ and~$e'_m$ are in \emph{immediate} conflict, by
       definition~$e_m$ has no conflicth with any event in $\causes{e'_m}$.
       Consider now any event $\tilde e \in X_1$.
       Observe that $\tilde e \in \hat C$.
       If $\tilde e \in \fcfl{e_m}$, then by \cref{e:existsj} we have that
       $\tilde e \in \fcfl{e_i}$, which implies that $\tilde e \notin X_1$.
       So~$e_m$ has no conflict with any event in~$X_1$.
     \item
       \emph{Fact 3:
       it holds that $C_m \cup \set{e_m} \subseteq X_3$.}
       Since~$C_m \subseteq C$, by \cref{e:step2},
       and~$C \subseteq X_1 \subseteq X_3$, we clearly have
       that~$C_m \subseteq X_3$.
       Also, $e_m \in X_3$ by definition.
     \item
       \emph{Fact 4:
       it holds that $X_3 \cap D_m = \emptyset$.}
       By \cref{e:step1,e:basic4} we know that
       $D_m \subseteq D \subseteq \ex C$.
       Since the sets~$\en C$ and~$\cex C$ partition~$\ex C$ we make the
       following argument.
       For any $\tilde e \in D_m \cap \cex C$ we know that
       $\tilde e \notin X_3$, as $C \subseteq X_3$.
       As for $D_m \cap \en C$ we have that
       $D_m \cap \en C = \set{e_1, \ldots, e_{m-1}}$.
       So for any $j \in \set{1, \ldots, m-1}$, because of the minimality
       of~$m$, we know that
       $\fcfl{e_i} \cap \hat C \supseteq \fcfl{e_j} \cap \hat C$
       does not hold.
       In other words, we know that there exists at least one event
       $\tilde e \in \fcfl{e_j} \cap \hat C$
       such that $\tilde e \notin \fcfl{e_i} \cap \hat C$.
       This implies that $\tilde e \notin \fcfl{e_i}$, and as a result
       $\tilde e \in X_1 \subseteq X_3$.
       So, for any event in~$D_m$ there is at least one conflicting event
       in~$X_3$, and~$X_3$ is a configuration.
       Therefore $X_3 \cap D_m = \emptyset$.
     \end{itemize}
     To show that $X_3$ has been explored in the subtree rooted at~$b_m$,
     consider the left
     child~$b'_m \eqdef \tup{C_m \cup \set{e_m}, D_m, \cdot, \cdot}$
     of~$b_m$.
     The induction hypothesis applies to any node~$\hat b \in B$
     in the subtree rooted at $b'_m$ (\ie, $b'_m \calls^* \hat b$).
     This is because $\hat b \lessdot b'_m \lessdot b_m \lessdot b$.
     By the first two facts previously proved, we know that $X_3$ is a
     configuration. The last two facts, together with the fact that
     the induction hypothesis holds on the subtree rooted at~$b'_m$,
     imply, by \cref{l:compl.findit},
     that some maximal configuration $C'' \supseteq X_3$
     has been explored in the subtree rooted at~$b'_m$.
     Since $e_m \in X_3$ and $e'_m \in \cex{X_3} \subseteq \cex{C''}$,
     we know that~$e'_m$ have been discovered at least when
     exploring~$C''$. Since $e_m \icfl e'_m$ and $e_m$ is in set~$D$
     we also know that \remove{$\cdot$} cannot remove $e'_m$ from~$U$
     before~$e_m$ is removed from~$D$.
     This implies that $e'_m \in U_m$, but also that
     $e'_m \in U_n$.

     Now, our goal was proving that $e'_i \in U_n$.
     Since $e'_m \in \fcfl{e_i}$, by \cref{e:existsj},
     there is some $\tilde e \in \ficfl[]{e_i}$ such
     that~$\tilde e \le e'_m$.
     Since~$U_n$ is causally closed, we have that
     $\tilde e \in U_n$.

     We have found some event $\tilde e \in U_n$ such that
     $e_i \icfl \tilde e$.
     If $\tilde e \ne e'_i$, then we substitute~$e'_i$ in~$\hat J$
     by~$\tilde e$.
     This means that in the definition of~$\hat J$ we cannot chose any
     arbitrary $e'_i$ from~$\hat C$
     (as we said before, to keep things simple).
     But we can always find at least one event in~$\hat C$ that is in
     immediate conflict with~$e_i$ and is also present in~$U_n$.
     Observe that the choice made for~$e_i$, with~$i \in \set{1, \ldots, n}$
     has no consequence for the choices made for~$j \in \set{1,\ldots,i-1}$.
     This means that we can always make a choice for index~$i$ after having
     made choices for every $j < i$.
   \end{itemize}

   This completes the argument showing that every~$e'_i$ (possibly
   modifying the original choice) is in~$\hat U$,
   and shows that $\hat J \subseteq \hat U$.
   This implies, by construction of $\hat J$, that
   $\hat J \in \alternatives{$C, D \cup \set e$}$ when the set of
   events~$U$ present in memory equals $\hat U$.
   As a result, \cref{a:a1} will do a recursive call at \cref{l:a1r}
   and~$b$ will have a right child. This is what we wanted to prove.
\end{proof}

\begin{lem}
   \label{l:compl.findit}
   For any node $b \eqdef \tup{C,D,\cdot,e} \in B$ in the call graph and any
   maximal configuration $\hat C \subseteq E$ of $\uunf$,
   if~$C \subseteq \hat C$ and~$D \cap \hat C = \emptyset$ and
   \cref{l:compl.step} holds on all nodes in the subtree rooted at~$b$,
   then there is a node~$b' \eqdef \tup{C',\cdot,\cdot,\cdot} \in B$ such that
   $b \calls^* b'$, and $\hat C = C'$.
\end{lem}
\begin{proof}
   Assume that \cref{l:compl.step} holds on any node~$b'' \in B$ such that
   $b \calls^* b''$, \ie, all nodes in the subtree rooted at~$b$.
   Since $C \subseteq \hat C$ and $D \cap \hat C = \emptyset$, we can apply
   \cref{l:compl.step} to $b$ and $\hat C$.
   If~$C$ is maximal, then clearly $C = \hat C$ and we are done.
   If not we consider two cases.
   If $e \in \hat C$, then by \cref{l:compl.step} we know that~$b$ has a
   left child~$b_1 \eqdef \tup{C_1, D_1, \cdot, e_1}$,
   with $C_1 \eqdef C \cup \set e$ and $D_1 \eqdef D$.
   Finally, if $e \notin \hat C$, then equally by \cref{l:compl.step} we know
   that~$b$ has a right
   child~$b_1 \eqdef \tup{C_1, D_1, \cdot, e_1}$,
   with $C_1 \eqdef C$ and $D_1 \eqdef D \cup \set e$.
   Observe, in any case, that $C_1 \subseteq \hat C$ and
   $D_1 \cap \hat C = \emptyset$.

   If $C_1$ is maximal, then necessarily $C_1 = \hat C$, we
   take~$b' \eqdef b_1$ and we have finished.
   If not, we can reapply \cref{l:compl.step} at~$b_1$ and
   make one more step into one of the children $b_2$ of~$b_1$.
   If $C_2$ still not maximal (thus different from $\hat C$)
   we need to repeat the argument starting from~$b_2$ only a finite
   number~$n$ of times until we reach a
   node~$b_n \eqdef \tup{C_n, D_n, \cdot, \cdot}$ where $C_n$ is a maximal
   configuration.
   This is because every time we repeat the argument on a non-maximal
   node~$b_i$ we advance one step down in the call tree, and all paths in
   the tree are finite. So eventually we find a leaf node~$b_n$
   where $C_n$ is maximal and satisfies $C_n \subseteq \hat C$.
   This implies that $C_n = \hat C$, and we can take $b' \eqdef b_n$.
\end{proof}

\thmcompleteness*
\begin{proof}
   We need to show that for every maximal
   configuration~$\hat C \subseteq E$ we can find a node
   $b \eqdef \tup{C, \cdot, \cdot, \cdot}$ in~$B$
   such that $\hat C = C$.
   This is a direct consequence of \cref{l:compl.findit}.
   Consider the root node of the tree,
   $b_0 \eqdef \tup{C,D,A,\bot}$, where
   $C = \set \bot$ and $D = A = \emptyset$.
   Clearly $C \subseteq \hat C$ and
   $D \cap \hat C = \emptyset$, and
   \cref{l:compl.step} holds on all nodes of the call tree.
   So \cref{l:compl.findit} applies to $\hat C$ and $b_0$,
   and it establishes the existence of the aforementionned node~$b$.
\end{proof}

\subsection{Memory Consumption}
\label{x:memory.lazy}

The following proposition establishes that \cref{a:a1} cleans set~$U$
adequately, and that after finishing the execution of \explore{$C,D,A$}, set~$U$
has the form described by the proposition.

\begin{prop}
   Assume the function \explore{$C,D,A$} is eventually called.
   Let~$\tilde U$ and~$\hat U$ be, respectively, the values of set~$U$ in
   \cref{a:a1} immediately before and immediately after executing the call.
   If $Q_{C,D,\tilde U} \subseteq \tilde U
   \subseteq Q_{C,D,\tilde U} \cup \en C$, then
   $\hat U = Q_{C,D,\hat U}$.
\end{prop}
\begin{proof}
   Let $b \eqdef \tup{C,D,A,e} \in B$ be the node in the call tree
   associated to the call to \explore{$C,D,A$}.
   The proof is by induction on the length of the longest path to a leaf
   starting from~$b$ (in the subtree rooted at~$b$).

   \emph{Base case.}
   The length is~0, $b$ is leaf node, and $C$ is a maximal configuration.
   Then $\en C = \emptyset$, so $\tilde U \subseteq Q_{C,D,\tilde U}$.
   By hypothesis $Q_{C,D,\tilde U} \subseteq \tilde U$ also holds,
   so $\tilde U = Q_{C,D,\tilde U}$.
   Now, the call to \extend{$C$} adds to $U$ only events from $\cex C$.
   So at \cref{l:a1ret}, clearly $\hat U = Q_{C,D,\hat U}$.

   \emph{Step case.}
   Let $U_1 \eqdef \tilde U$ be the value of set~$U$ immediately before
   the call to the function
   \explore{$C,D,A$}.
   Let $U_2$ be the value immediately before \cref{a:a1} makes the first
   recursive call, at \cref{l:a1l};
   $U_3$ the value immediately after that call returns;
   $U_4$ immediately after the second recursive call returns;
   and $U_5 \eqdef \hat U$ immediately after the call to \explore{$C,D,A$}
   returns.
   Assume that $Q_{C,D,U_1} \subseteq U_1 \subseteq Q_{C,D,U_1} \cup \en C$
   holds.
   Let $C' \eqdef C \cup \set e$.
   We first show that
   \[
   Q_{C',D,U_2} \subseteq U_2 \subseteq Q_{C',D,U_2} \cup \en{C'}
   \]
   holds.
   This ensures that the induction hypothesis applies to the first
   recursive call, at \cref{l:a1l}, and guarantees that
   $U_3 = Q_{C', D, U_3}$.

   Let $\tilde e$ be an event in $Q_{C',D,U_2}$.
   We show that $\tilde e \in U_2$.
   First, remark that $U_2 = U_1 \cup \ex C$.
   If $\tilde e \in C \cup D \subseteq U_1 \subseteq U_2$, we are done.
   If $\tilde e = e$, then clearly $\tilde e \in \ex C \subseteq U_2$.
   Otherwise $\tilde e$ is in $[e_1]$ for some $e_1 \in U_2$ such that
   there is some $e_2 \in C' \cup D$ with $e_1 \icfl e_2$.
   Since celarly $U_2$ is causally closed and $e_1 \in U_2$, we have that
   $\tilde e \in U_2$.

   Let $\tilde e$ be now an event in $U_2$.
   We show that $\tilde e \in Q_{C',D,U_2} \cup \en{C'}$.
   If $\tilde e \in U_1$, the clearly $\tilde e \in Q_{C',D,U_2}$
   (esentially because $U_1 \subseteq U_2$).
   So assume that $\tilde e \in U_2 \setminus U_1 = \ex C$.
   Now, observe that $\ex C \subseteq \set e \cup \ex{C'}$.
   We are done if $\tilde e \in \set e \cup \en{C'}$,
   so assume that $\tilde e \in \cex{C'}$.
   Since $C' \subseteq U_2$ and $\tilde e \in U_2$,
   by definition we have $\tilde e \in Q_{C',D,U_2}$.
   This shows that $\tilde e \in Q_{C',D,U_2} \cup \en{C'}$.

   Then by induction hypothesis we have that
   $U_3 = Q_{C', D, U_3}$ immediately after the recursive call of
   \cref{l:a1l} returns.
   Function \alternatives{$\cdot$} does not update~$U$,
   so when the second recursive call is made, \cref{l:a1r}, clearly
   \[
   Q_{C,D',U_3} \subseteq U_3 \subseteq Q_{C,D',U_3} \cup \en{C}
   \]
   holds, with $D' \eqdef D \cup \set e$.
   This is obvious after realizing the fact that
   \[
   Q_{C \cup \set e,D,U_3} = Q_{C,D \cup \set e,U_3}.
   \]
   So the induction hypothesis applies to the second recursive call as
   well, and guarantees that
   $U_4 = Q_{C, D \cup \set e, U_4}$ holds
   immediately after the recursive call of \cref{l:a1r} returns.

   Recall that our goal is proving that
   $U_5 = Q_{C, D, U_5}$.
   The difference between $U_4$ and $U_5$ are the events removed by the
   call to the function \remove{$e,C,D$}.
   Let~$R$ be such events (see below for a formal definition).
   Then we have that $U_5 = U_4 \setminus R$.
   In the sequel we show that the following equalities hold:
   \begin{equation}
   \label{e:u5}
   U_5 =
   U_4 \setminus R = 
   Q_{C,D \cup \set e,U_4} \setminus R =
   Q_{C,D,U_4} =
   Q_{C,D,U_5}
   \end{equation}
   Observe that these equalities prove the lemma.
   In the rest of this proof we prove the various equalities above.

   To prove \cref{e:u5}, first observe that the events
   removed from~$U$ by
   \remove{$e, C, D$}, called $R$ above, are exactly
   \begin{equation}
   \label{e:rem}
   R \eqdef
   \left ( \set e \cup \bigcup_{e' \in \ficfl[U_4] e} [e'] \right )
   \setminus Q_{C,D,U_4}.
   \end{equation}
   This is immediate from the definition of \remove{$\cdot$}.
   Now we prove two statements, \cref{e:u4r} and~\cref{e:u45},
   that imply the validity of~\cref{e:u5}.
   We start stating the first:
   \begin{equation}
   \label{e:u4r}
   Q_{C,D \cup \set e,U_4} \setminus R = Q_{C,D,U_4}.
   \end{equation}
   This equality intuitively says that (left-hand side) executing
   \remove{$e,C,D$} when the set $U$ contains the events in $U_4$
   (remember that $U_4 = Q_{C,D \cup \set e,U_4}$) leaves in $U$
   exactly (right-hand side)
   all events in $C$, all events in $D$, and all events that causally
   precede some other event from $U$ (in fact, $U_4$) which is is conflict
   with some event in $C \cup D$.
   For the shake of clarity, unfolding the definitions in \cref{e:u4r}
   yields the following equivalent equality:
   \[
   \left (
   C \cup D \cup \set e \cup
   \bigcup_{\substack{e' \in C \cup D \cup \set e \\
   e'' \in \ficfl[U_4]{e'} }}
   [e'']
   \right )
   \setminus
   \left (
   \left (
   \set e \cup \bigcup_{e' \in \ficfl[U_4] e} [e'] 
   \right )
   \setminus Q_{C,D,U_4}
   \right )
   =
   Q_{C,D,U_4}
   \]
   We now prove \cref{e:u4r}.
   Let $\tilde e$ be an event contained in the left-hand side.
   We show that $\tilde e$ is in $Q_{C,D,U_4}$.
   We are done if $\tilde e \in C \cup D$.
   If $\tilde e = e$, then $\tilde e \notin R$.
   Now, from the definition \cref{e:rem} of~$R$ we get that
   $\tilde e \in Q_{C,D,U_4}$.
   Lastly, if $\tilde e \notin C \cup D \cup \set e$,
   then there is some event $e' \in C \cup D \cup \set e$
   and some event $e'' \in U_4$ such that
   $e' \icfl e''$ and $\tilde e \le e''$.
   If $e' \in C \cup D$, then by defnition
   $\tilde e \in Q_{C,D,U_4}$.
   The case that $e' = e$ cannot happen, as we show now.
   Since $\tilde e$ is in the left-hand side, $\tilde e$ is not in~$R$.
   If $\tilde e \notin R$,
   then $\tilde e$ is either in $Q_{C,D,U_4}$, as we wanted to show,
   or $\tilde e$ is not in
   $\set e \cup \bigcup_{\hat e \in \ficfl[U_4] e} [\hat e]$.
   This means that $e' \ne e$.

   For the opposite direction,
   let $\tilde e$ be an event in $Q_{C,D,U_4}$.
   We show that it is contained in the left-hand side set.
   By definition $\tilde e \notin R$.
   If $\tilde e \in C \cup D$, clearly $\tilde e$ is in the left-hand side.
   If not,
   then there is some event $e' \in C \cup D$
   and some event $e'' \in U_4$ such that
   $e' \icfl e''$ and $\tilde e \le e''$.
   Then by definition~$\tilde e$ is in the left-hand side.
   This completes the proof of~\cref{e:u4r}.

   The second statement necessary to prove \cref{e:u5} is the following:
   \begin{equation}
   \label{e:u45}
   Q_{C,D,U_4} =
   Q_{C,D,U_5}
   \end{equation}
   From left to right.
   Assume that $\tilde e \in Q_{C,D,U_4}$.
   Routinary if $\tilde e \in C \cup D$.
   Assume otherwise that there is some $e_1 \in C \cup D$
   and~$e_2 \in \ficfl[U_4]{e_1}$ such that $\tilde e \in [e_2]$.
   We show that $e_2 \in U_5$, which clearly proves that
   $\tilde e \in Q_{C,D,U_5}$.
   By definition $e_2 \in U_4$.
   By \cref{e:rem}, clearly $e_2 \notin R$, as $e_2 \in Q_{C,D,U_4}$.
   Since $U_5 = U_4 \setminus R$ we have that $e_2 \in U_5$.

   From right to left the proof is even simpler.
   Assume that $\tilde e \in Q_{C,D,U_4}$.
   Routinary if $\tilde e \in C \cup D$.
   Assume otherwise that there is some $e_1 \in C \cup D$
   and~$e_2 \in \ficfl[U_5]{e_1}$ such that $\tilde e \in [e_2]$.
   Since $U_5 \subseteq U_4$, clearly $e_2 \in U_4$ and so
   $e_2 \in Q_{C,D,U_4}$.
   Then $\tilde e \in Q_{C,D,U_4}$ as the latter is causally closed.
\end{proof}


\section{Proofs: Improvements}
\label{x:impro}

\subsection{Completeness with Cutoffs}
\label{x:cutoffs}

\newcommand\pred{\ensuremath{\mathcal{P}_1}}
\newcommand\pblue{\ensuremath{\mathcal{P}_2}}

In \cref{s:cutoffs} we describe a modified version of \cref{a:a1}, where
the \extend procedure has been replaced by the \extendcut procedure.
The updated version uses a predicate $\iscutoff{e,U,G}$ to decide when an event
is added to~$U$.
We refer to this version as the \emph{updated algorithm}.

Like \cref{a:a1}, the updated algorithm also explores a binary tree.
It works by, intuitively, ``\emph{allowing}'' \cref{a:a1} to ``\emph{see}'' only
the non-cutoff events. The terminal configurations it will explore,
\ie, those at which the procedure \ena{$C$} of \cref{a:a1} returns an empty set,
will be those for which any enabled event in $\en C$ has been declared a cutoff.

Many properties remain true in the updated algorithm, \eg, \cref{l:general}.
Consider the set of terminal configurations explored by the updated
algorithm, and let us denote them by
\[
C_1, C_2, \ldots, C_n.
\]
Let $\ppref' \eqdef \tup{E', <', {\cfl'}}$ be the unique prefix of~$\uunf$
whose set of events~$E'$ equals $\bigcup_{1 \le i \le n} C_i$.
Whenever \cref{a:a1} is applied to an acyclic state-space (all executions
terminate), the following properties hold:
\begin{itemize}
\item
  $\ppref' = \uunf$;
\item
  Each configuration $C_i$ is a maximal configuration of~$\ppref'$.
\end{itemize}
However, when we apply the updated algorithm to an arbitrary system (with
possibly non-terminating executions), none of these properties remain valid in
general.
Obviously the first one will not be valid, \eg, if $\uunf$ is infinite, this was
expected and intended.
The second property will also not be valid in general, essentially because one
event could be declared as cutoff when exploring one configuration and as
non-cutoff when exploring another configuration. We illustrate this with an
abstract example.

\begin{exampl}
Assume that $\uunf$ is infinite and has only two maximal (infinite) configurations.
The updated algorithm will explore the first until reaching some
first terminal (and finite) configuration $C_1$ where all events in $\en{C_1}$
have been declared as cutoffs.
Let $e$ be one of those cutoffs in $\en{C_1}$, and $e'$ the corresponding event
in~$U \cup G$.
The algorithm will then backtrack, and start exploring the second configuration.
It could then very well reach a configuration that enables $e$.
The updated algorithm will have to re-decide whether $e$ is a cutoff.
If it decides that it is not, \eg, because the corresponding event $e'$ has been
discarded from $U \cup G$, it could add $e$ to~$C$, and so the second maximal
configuration $C_2$ explored in this way will contain some event enabled
by~$C_1$.
This implies that $C_1$ is not a maximal configuration of the prefix $\ppref'$.
\end{exampl}

This means essentially that proving that $\ppref'$ is a \emph{complete
prefix}~\cite{ERV02} is not a valid strategy
for proving \cref{t:cutoff.completeness},
since potentially there exists configurations $C$ of $\ppref'$ such that
$C \not\subseteq C_i$ for any $1 \le i \le n$.

Alternatively, we could try to reason using a variant of McMillan's standard
argument~\cite{Mcm93,ERV02,BHKTV14} (largely used in the literature about
unfoldings for proving that some unfolding prefix is complete).
Given a state $s \in \reach M$, we want to show that there is some
configuration $C$ such that
\begin{equation}
\label{e:mcmfails}
\state C = s \text{ and } C \subseteq C_i \text{ for some } 1 \le i \le n.
\end{equation}
We know that $\uunf$ contains some configuration $C'$ such that $\state{C'} =
s$. If $C'$ satisfies \cref{e:mcmfails} we are done.
If not, the usual argument now finds that $C'$ has a cutoff event,
but this does not work in our context:
we can easily show that some maximal configuration of $\ppref'$ enables
some event in $C'$ but not in $\ppref'$ (the wished cutoff),
but there is no guarantee that that maximal configuration is one of the $C_i$'s
above, so there is no guarantee that the updated algorithm has explicitly
declared that event as cutoff.

As a result, we resort to a completely different argument.
The main idea is simple. We divide the set of events in $\ppref'$
in two parts, the \emph{red} events and the \emph{blue} events.
Red events are such that the updated algorithm never declares them cutoff,
blue events have at least been declared once cutoff and once non-cutoff.
We next show two things.
First, that the red events contain one representative configuration for every
reachable marking (contain a complete prefix).
Second, that every configuration formed by red events has been explored by the
updated algorithm.
Together, these implies \cref{t:cutoff.completeness}.

We start with two definitions.
\begin{itemize}
\item
  Let the \emph{red prefix} be the unique prefix
  $\pred \eqdef \tup{E_1, {<}, {\cfl}}$
  of~$\uunf$
  formed by those events $e$ added at least once to~$U$ by the updated algorithm
  and such that every time \extendcut evaluated the predicate
  $\iscutoff{e,U,G}$, the result was \emph{false}.
\item 
  Let the \emph{blue prefix} be the unique prefix
  $\pblue \eqdef \tup{E_2, {<}, {\cfl}}$
  of~$\uunf$
  such that $E \eqdef \bigcup_{1 \le i \le n} C_i$.
\end{itemize}
Observe that $\pblue$ is in fact what we called $\ppref'$ so far.
Notice also that $E_1 \subseteq E_2$.

In \cref{s:cutoffs} we defined the $\iscutoff \cdot$ predicate using
McMillan's size order. Here we redefine it to use an arbitrary \emph{adequate
order}. This allows us to prove a more general version of \cref{t:cutoff.completeness}.
Let $\prec$ be an adequate order (we skip the definition,
the interested reader can find it in~\cite{ERV02}) on the configurations
of~$\uunf$.
We define~$\iscutoff{e,U,G}$ to hold iff
there exists some event $e' \in U \cup G$ such that
\begin{equation}
\label{e:cutoff.adeq}
\state{[e]} = \state{[e']}
\text{ ~ and ~ }
[e'] \prec [e].
\end{equation}
The \emph{size order} from McMillan, which we used in \cref{s:cutoffs} is indeed
adequate~\cite{ERV02}.

We now need to define the canonical prefix associated with~$\prec$
(we refer the reader to~\cite{EH08}, to avoid increasing the limited space in
the References section, although a better reference would be
[Khomenko, Koutny, Vogler 2002]).
We give a simplified definition.
Given a event $e \in E$, we call it \emph{$\prec$-cutoff} iff
there exists some other event $e' \in E$ such that \cref{e:cutoff.adeq} holds.
Observe that we now search $e'$ in~$E$ and not in~$U \cup G$.
The \emph{$\prec$-prefix} is the unique $\ispref$-maximal unfolding prefix that
contains no \mbox{$\prec$-cutoff}.
It is well known~\cite{EH08} that,
(1) the $\prec$-prefix exists and is unique,
(2) it is \emph{marking-complete}, \ie, for every $s \in \reach M$, there is some
configuration $C$ in \mbox{$\prec$-cutoff} such that $\state C = s$.

The key observation now is that all events in $\prec$-prefix are red,
\ie, the $\prec$-prefix is a prefix of $\pred$.
Clearly,
regardless of the actual contents of~$U$ and~$G$ when $\iscutoff{e,U,G}$ is
evaluated, the result will always be \emph{false}
if~$e$ is not \mbox{$\prec$-cutoff}.

So, in order to prove \cref{t:cutoff.completeness}, it suffices to show that
every \emph{red} configuration from $\pred$ is contained in some node explored
the algorithm. We achieve this with \cref{l:cutoff.step} and
\cref{l:cutoff.findit}.

\begin{lem}
   \label{l:cutoff.step}
   Let $b \eqdef \tup{C,D,A,e} \in B$ be a node in the call graph and
   $\hat C \subseteq E_1$ an arbitrary \emph{red} configuration in $\pred$,
   such that the following two conditions are verified:
   \begin{enumerate}
   \item
     $C \cup \hat C$ is a configuration, and
   \item
     for any $\tilde e \in D$ there is some
     $e' \in \hat C$ such that $\tilde e \icfl e'$.
   \end{enumerate}
   Then exactly one of the following statements hold:
   \begin{itemize}
   \item
     Either $b$ is a leaf node in $B$, or
   \item
     for any $\hat e \in \hat C$ we have $\lnot (e \icfl \hat e)$
     and~$b$ has a left child, or
   \item
     for some $\hat e \in \hat C$ we have $e \icfl e$
     and~$b$ has a right child.
   \end{itemize}
\end{lem}
\begin{proof}
   The statement of this lemma is very similar to the one of
   \cref{l:compl.step}, the main lemma behind the proof of
   \cref{t:a1.completeness} (completeness).
   Consequently the proof is also similar.
   The proof is by induction on~$b$ using the same total order
   ${\lessdot} \in B \times B$ that we employed for \cref{l:compl.step}.

   \emph{Base case.}
   Node~$b$ is the least element in~$B$ \wrt $\lessdot$.
   It is therefore the leftmost leaf of the call tree.
   Then the first item holds.

   \emph{Step case.}
   Assume that the result holds for any node $\tilde b \lessdot b$.
   If~$C$ is maximal, we are done. 
   So assume that $C$ is not maximal.
   Then~$b$ has at least one left child.
   If we can find some $\hat e \in \hat C$ such that
   $\hat e \icfl e$, then the second item holds and we are done.
   
   So assume that that for some
   $\hat e \in \hat C$ it holds that $\hat e \icfl e$.
   We show that the third item holds in this case.
   For that we need to show that $b$ has a right child.
   The rest of this proof accomplishes that,
   it shows that there is some alternative
   $J \in \alternatives{$C, D \cup \set e$}$ whenever the algorithm asks
   for the existence of one.

   We define the set
   \[
   F \eqdef \set{e_1, \ldots, e_n} \eqdef D \cup \set e.
   \]
   This set contains the events that
   the alternative~$J$ needs to \emph{justify}.
   Let $e_i$ be any event in~$F$.
   By hypothesis there exists some $e'_i \in \hat C$ such that
   $e_i \icfl e'_i$.
   Thus, there exists at least one set
   \[
   J \eqdef [\set{e'_1, \ldots, e'_n}]
   \]
   where
   $e'_i \in \hat C$ and $e_i \icfl e'_i$
   for $i \in \set{1, \ldots, n}$.
   Clearly, $J \subseteq \hat C$ and so it is a red configuration of~$\pred$.
   We remark that~$J$ is not uniquely defined,
   there may be several~$e'_i$ to choose for each~$e_i$.
   For now, take any suitable $e'_i$ without further regard.
   We will later refine this choice if necessary.

   We show now that $J \in \alternatives{$C, D \cup \set e$}$
   when function \alternatives{$\cdot$} is called just before
   \cref{l:a1r} during the execution of \explore{$C,D,A$}.
   Let $\hat U$ be the set of events contained in the variable~$U$
   exactly when \alternatives{$\cdot$} is called.

   By construction $J \cup C$ is configuration, and contains an event in
   conflict with any event in $D \cup \set e$.
   We only need to check that~$J \subseteq \hat U$, \ie,
   that all events in~$J$ were are known (in fact, \emph{remembered})
   when function \alternatives{$\cdot$} is called.

   We reason about the call stack when the algorithm is situated at
   $b = \tup{C,D,A,e}$.
   For~$i \in \set{1, \ldots, n}$ let
   $b_i \eqdef \tup{C_i, D_i, A_i, e_i} \in B$ be the node in the call
   graph associated to event $e_i \in F$.
   These nodes are all situated in the unique path from~$b_0$ to~$b$.
   \Wlog assume (after possible reordering of the index~$i$) that
   \[
   b_0 \calls^*
   b_1 \calls^*
   b_2 \calls^*
   \ldots \calls^*
   b_n,
   \]
   where $b_n = b$ and $e_n = e$.
   Since every event $e_i$ is in $D = D_n$,
   for $i \in \set{1, \ldots, n-1}$,
   we know that the first step in the path that goes
   from $b_i$ to $b_{i+1}$ is a \emph{right child}.
   Also, we remark that by construction we have
   $\set{e_1, \ldots, e_{i-1}} = D_i$ for every 
   $i \in \set{2, \ldots, n}$.

   We need to show that $e'_i \in \hat U$, for $i \in \set{1,\ldots,n}$.
   We consider two cases.
   Consider the set $D_i = \set{e_1, \ldots, e_{i-1}}$.
   Only two things are possible:
   either there exists some $j \in \set{1, \ldots, i-1}$ such that
   \begin{equation}
   \label{e:ej}
   \fcfl{e_j} \cap \hat C \subseteq \fcfl{e_i} \cap \hat C
   \end{equation}
   holds, or for all $j \in \set{1, \ldots, i-1}$ the above statement is false.

   \begin{itemize}
   \item
     \emph{Case 1}:
     for all $j \in \set{1, \ldots, i-1}$ we have that \cref{e:ej} do not hold.
     This means that for all such~$j$,
     some event in $\fcfl{e_j} \cap \hat C$
     is not in $\fcfl{e_i} \cap \hat C$.
     Consider the set
     \[
     X_1 \eqdef \hat C \setminus \fcfl{e_i}.
     \]
     It is a red configuration of $\pred$, which satisfies the following
     properties:
     \begin{itemize}
     \item
       \emph{Fact 1: set $X_1 \cup C_i \cup \set{e_i}$ is a configuration.}
       Since $X_1 \cup C_i \subseteq \hat C \cup C$,
       clearly $X_1 \cup C_i$ is a configuration.
       Also, $X_1$ has no event in conflict with $e_i$ by construction.
     \item
       \emph{Fact 2: for any $\tilde e \in D_i$ there is some
       $e' \in X_1$ such that $\tilde e \icfl e'$.}
       This holds by construction.
       For any $\tilde e \in D_i = \set{e_1, \ldots, e_{i-1}}$ we know that
       some event in $\fcfl{\tilde e} \cap \hat C$ is not in
       $\fcfl{e_i} \cap \hat C$, so it is necessarily in $X_1$.
     \end{itemize}
     Consider the left
     child~$b'_i \eqdef \tup{C_i \cup \set{e_i}, D_i, \cdot, \cdot}$
     of~$b_i$.
     Every node $\hat b$ in the subtree rooted at $b'_i$
     (\ie, $b'_i \calls^* \hat b$)
     is such that $\hat b \lessdot b_i \lessdot b$.
     The induction hypothesis thus applies to~$\hat b$.
     By the previous facts,
     \cref{l:cutoff.findit} applied to~$b'_i$ and~$X_1$ implies that
     some leaf (maximal) configuration
     $C' \supseteq X_1$ has been explored in the subtree rooted at~$b'_i$.
     Since $e'_i$ is a red event (it will never be declared cutoff)
     and $e'_i \in \ex{C'}$, event $e'_i$ will be discovered when
     exploring~$C'$, and will be kept in~$U$ as long as~$e_i$ remains
     in~$U$. As a result $e'_i \in \hat U$, which we wanted to prove.
   \item
     \emph{Case 2}:
     there is some $j \in \set{1, \ldots, i-1}$ such that \cref{e:ej}
     holds.
     Let~$m$ be the minimum such integer.
     Consider the set $X_2$ defined as
     \[
     X_2 \eqdef \hat C \setminus \fcfl{e_i} \cup \causes{e'_m}
     \]
     It is clearly a subset of $\hat C$, so it is a red configuration of
     $\pred$, and it satisfies the following properties:
     \begin{itemize}
     \item
       \emph{Fact 3: set $X_2 \cup C_m \cup \set{e_m}$ is a configuration.}
       Since $X_2 \cup C_m \subseteq \hat C \cup C$,
       clearly $X_2 \cup C_m$ is a configuration.
       Also, $X_2$ has no event in conflict with $e_m$,
       since all such events are in
       $\fcfl{e_i}$ and we have removed them.
       Observe that by adding $\causes{e'_m}$ we do no add any conflict,
       as there is no conflict between $e_m$ and any event of
       $\causes{e'_m}$.
     \item
       \emph{Fact 4: for any $\tilde e \in D_m$ there is some
       $e' \in X_2$ such that $\tilde e \icfl e'$.}
       This holds by construction, as a result of the minimality of~$m$.
       For any $\tilde e \in D_m = \set{e_1, \ldots, e_{i-m}}$ we know that
       \cref{e:ej} do not hold for $\tilde e$.
       So some event in $\fcfl{\tilde e} \cap \hat C$ is not in
       $\fcfl{e_i} \cap \hat C$, and so it is necessarily in $X_2$.
     \end{itemize}
     Like before, consider now the left
     child~$b'_m \eqdef \tup{C_m \cup \set{e_m}, D_m, \cdot, \cdot}$
     of~$b_m$.
     The induction hypothesis applies to any node~$\hat b \in B$
     in the subtree rooted at $b'_m$ (\ie, $b'_m \calls^* \hat b$).
     By the previous facts,
     \cref{l:cutoff.findit} applied to~$b'_m$ and~$X_2$ implies that
     some leaf (maximal) configuration
     $C' \supseteq X_2$ has been explored in the subtree rooted at~$b'_m$.
     Since $e'_m$ is a red event (it will never be declared cutoff)
     and $e'_m \in \ex{C'}$, event $e'_m$ will be discovered when
     exploring~$C'$, and will be kept in~$U$ as long as~$e_m$ remains
     in~$U$. As a result $e'_m \in \hat U$.

     We actually wanted to prove that $e'_i$ is in $\hat U$.
     This is now easy.
     Since $e'_m \in \fcfl{e_i}$, by \cref{e:ej},
     there is some $\tilde e \in \ficfl[]{e_i}$ such
     that~$\tilde e \le e'_m$.
     Since~$\hat U$ is causally closed, we have that
     $\tilde e \in \hat U$.

     We have found some event $\tilde e \in \hat U$ such that
     $e_i \icfl \tilde e$.
     If $\tilde e \ne e'_i$, then we substitute~$e'_i$ in~$J$
     by~$\tilde e$.
     This means that in the definition of~$J$ we cannot chose any
     arbitrary $e'_i$ from~$\hat C$
     (as we said before, to keep things simple).
     But we can always find at least one event in~$\hat C$ that is in
     immediate conflict with~$e_i$ and is also present in~$\hat U$.
     Observe that the choice made for~$e_i$, with~$i \in \set{1, \ldots, n}$
     has no consequence for the choices made for~$j \in \set{1,\ldots,i-1}$.
     This means that we can always make a choice for index~$i$ after having
     made choices for every $j < i$.
   \end{itemize}

   This completes the argument showing that every~$e'_i$ (possibly
   modifying the original choice) is in~$\hat U$,
   and shows that $J \subseteq \hat U$.
   This implies, by construction of~$J$, that
   $J \in \alternatives{$C, D \cup \set e$}$ when the set of
   events~$U$ present in memory equals $\hat U$.
   As a result, the algorithm will do a \emph{right recursive call} 
   and~$b$ will have a right child. This is what we wanted to prove.
\end{proof}

\begin{lem}
   \label{l:cutoff.findit}
   Let $b \eqdef \tup{C,D,\cdot,e} \in B$ be any node the call graph.
   Let $\hat C \subseteq E_1$ be any configuration of $\pred$, \ie,
   consisting \emph{only} of red events. Assume that
   \begin{itemize}
   \item
     $C \cup \hat C$ is a configuration;
   \item
     for any $\tilde e \in D$ there is some
     $e' \in \hat C$ such that $\tilde e \icfl e'$;
   \item
     \cref{l:cutoff.step} holds on every node in the subtree rooted at~$b$.
   \end{itemize}
   Then there exist in~$B$ a node~$b' \eqdef \tup{C',\cdot,\cdot,\cdot}$
   such that $b \calls^* b'$ and $\hat C \subseteq C'$.
\end{lem}
\begin{proof}
   Assume that \cref{l:cutoff.step} holds on any node~$b'' \in B$ such that
   $b \calls^* b''$, \ie, all nodes in the subtree rooted at~$b$.
   By hypothesis we can apply \cref{l:cutoff.step} to~$b$ and $\hat C$.
   If~$C$ is maximal, \ie, the algorithm do not find any non-cutoff
   extension of~$C$, then we have that $\hat C \subseteq C$, as otherwise
   any event in $\hat C \setminus C$ would be non-cutoff (as it is red)
   and would be enabled at $C$ (because $\hat C \cup C$ is a
   configuration).  So if~$b$ is a leaf, then we can take $b' \eqdef b$.
   
   If not, then~$e$ is enabled at~$C$ and there is at least a left child.
   Two things can happen now.
   Either~$e$ is in conflict with some event in~$\hat C$ or not.

   If~$e$ is not in conflict with any event in~$\hat C$, then the
   left child~$b_1 \eqdef \tup{C_1, D_1, \cdot, e_1}$,
   with $C_1 \eqdef C \cup \set e$ and $D_1 \eqdef D$,
   is such that $C_1 \cup \hat C$ is a configuration, and
   $\hat C$ contains some event in conflict with every event in $D_1$.
   Furthermore \cref{l:cutoff.step} applies to~$b_1$ as well.

   If $e$ is in conflict with some event in~$\hat C$,
   then by \cref{l:cutoff.step} we know
   that~$b$ has a right
   child~$b_1 \eqdef \tup{C_1, D_1, \cdot, e_1}$,
   with $C_1 \eqdef C$ and $D_1 \eqdef D \cup \set e$.
   Like before, $C_1 \cup \hat C$ is a configuration and
   for any event in~$D_1$ we have another one in $\hat C$ in conflict with
   it.

   In any case, if $C_1$ is maximal, then it holds that $\hat C \subseteq
   C_1$ and we are done.
   If not, we can reapply \cref{l:cutoff.step} at~$b_1$ and
   make one more step into one of the children $b_2$ of~$b_1$.
   If $C_2$ still do not contain $\hat C$, then
   we need to repeat the argument starting from~$b_2$ only a finite
   number~$n$ of times until we reach a
   node~$b_n \eqdef \tup{C_n, D_n, \cdot, \cdot}$ where $b_n$ has no
   further children in the call tree (\ie,
   $\en{C_n}$ is either empty or contains only cutoff events).
   This is because every time we repeat the argument on a non-leaf
   node~$b_i$ we advance one step down in the call tree, and all paths in
   the tree are finite. So eventually we find a leaf node~$b_n$, which, as
   argued earlier, satisfies that $\hat C \subseteq C_n$,
   and we can take $b' \eqdef b_n$.
\end{proof}

\thmcutoffcompleteness*
\begin{proof}
   Let $\ppref$ be an unfolding prefix constructed with the classic
   saturation-based unfolding algorithm, using the standard cutoff strategy in
   combination with an arbitrary adequate order $\prec$:
   an event~$e$ is a \emph{classic-cutoff} if there is another event~$e'$
   in~$\unf M$ such that $\state{[e]} = \state{[e']}$
   and $[e'] \prec [e]$.
   By construction all events in $\ppref$ are red, so they are in $\pred$.

   Let $\tup{B, {\calls}}$ be the call tree associated with one execution
   of \cref{a:a1} retrofitted with the cutoff mechanism.
   Let $s \in \reach M$ be an arbitrary state of the system.
   Owing to the properties of $\ppref_M$~\cite{ERV02}, there is a
   configuration~$\hat C$ in $\ppref$ such that~$\state{\hat C} = s$.
   Such a configuration is in $\pred$.

   Now, \cref{l:cutoff.findit} applies to the initial node~$b_0 \in B$
   and~$\hat C$, and guarantees that the algorithm will visit a node
   $b \eqdef \tup{C, \cdot, \cdot, \cdot} \in B$ such that
   such that $\hat C \subseteq C$. This is what we wanted to prove.
\end{proof}

\end{document}